\documentclass{article}
\usepackage{fullpage}

\usepackage{amsfonts,amssymb,amsthm,amsmath,amsopn}
\usepackage{graphicx, tikz, tikz-cd, epstopdf}

\usepackage{algorithmic}
\usepackage{bm,bbm}
\usepackage{hyperref}
\usepackage{verbatim}

\DeclareMathOperator{\Tr}{Tr}

\newcommand{\cW}{\mathcal{W}}

\newcommand{\ve}{a}

\newcommand{\bx}{\mathbf{x}}
\newcommand{\by}{\mathbf{y}}
\newcommand{\bk}{\mathbf{k}}
\newcommand{\bv}{\mathbf{v}}

\newcommand{\balpha}{\boldsymbol{\alpha}}
\newcommand{\bmu}{\bm{\mu}}

\newcommand{\nablax}{\nabla_{\bx}}
\newcommand{\nablak}{\nabla_{\mathbf{k}}}

\newcommand*{\im}{\mathop{}\!\mathrm{i}}
\newcommand{\ii}{\mathrm{i}}
\newcommand{\ee}{\mathrm{e}}
\newcommand{\dd}{\mathrm{d}}

\newcommand{\Id}{\mathbbm{1}}

\newcommand{\ol}{\overline}

\newcommand{\R}{\mathbb{R}}
\newcommand{\C}{\mathbb{C}}

\newcommand*{\Err}{\mathcal{E}}

\numberwithin{equation}{section}

\theoremstyle{plain} 
\newtheorem{theorem}{Theorem}[section]
\newtheorem*{theorem*}{Theorem}
\newtheorem{lemma}[theorem]{Lemma}
\newtheorem*{lemma*}{Lemma}
\newtheorem{corollary}[theorem]{Corollary}
\newtheorem*{corollary*}{Corollary}
\newtheorem{proposition}[theorem]{Proposition}
\newtheorem*{proposition*}{Proposition}

\theoremstyle{remark}

\newtheorem{remark}[theorem]{Remark}
\newtheorem*{remark*}{Remark}

\newtheorem*{remarks*}{Remarks}

\theoremstyle{definition} 
\newtheorem{definition}[theorem]{Definition}
\newtheorem*{definition*}{Definition}

\newtheorem*{assumption*}{Assumption}

\newcommand{\MSC}[1]{\par\noindent\textbf{Mathematics Subject Classification (2020):} #1\par}

\title{The Semiclassical Limit of the 2D Dirac--Hartree Equation\\ with Periodic Potentials}

\author{
Jinyeop Lee
\thanks{
	Department of Mathematics and Computer Science, University of Basel, Basel, Switzerland.
	Email: \href{mailto:jinyeop.lee@unibas.ch}{jinyeop.lee@unibas.ch}}
\and 
Kunlun Qi
\thanks{
	Simons Laufer Mathematical Sciences Institute (former MSRI), Berkeley, CA 94720, USA.\\
	\indent\indent Department of Computational Mathematics, Science and Engineering and Department of Mathematics, Michigan State University, East Lansing, MI 48824, USA.  
	Email: \href{mailto:kunlunqi.math@gmail.com}{kunlunqi.math@gmail.com}}
}

\date{\today}

\begin{document}

\maketitle

\begin{abstract}
	We study the semiclassical limit of the two-dimensional Dirac--Hartree equation in the presence of a periodic external potential. The spinor dynamics are formulated using the matrix-valued Wigner transform together with spectral projectors onto the positive and negative energy bands. Under suitable assumptions on the initial data and the potentials, we rigorously derive Vlasov-type transport equations describing the evolution of the band-resolved phase-space densities in both the massive and massless regimes. In the massless case, the limiting dynamics propagate ballistically with constant speed, while in the massive case the velocity is relativistic. Our analysis justifies the emergence of relativistic Vlasov equations from Dirac--Hartree dynamics in the semiclassical regime.
	As a corollary, we recover the relativistic Vlasov--Poisson equation from the Dirac equation with a regularized Coulomb interaction when the regularization vanishes together with the semiclassical parameter.
	
	\smallskip
	\MSC{35Q40; 81Q05; 82C10} 
\end{abstract}

\section{Introduction}\label{sec:intro}

\subsection{The model}

The two-dimensional (2D) Dirac equation with the external and Hartree-type interactive potential reads:
\begin{equation}\label{eq:massive-dirac}
	\im \hbar \partial_t \psi(t, \bx) = \Big( -\im \hbar c\, \boldsymbol{\alpha} \cdot \nabla + m c^2\,\gamma^0 + V_\Gamma(\bx)\Id  + V_{\mathrm{int}}(t,\bx) \Id \Big) \psi(t, \bx)
\end{equation}
for the initial condition,
\begin{equation}\label{eq:initial-psi}
	\psi(t=0, \bx) = \psi_0(\bx) \in L^2(\mathbb{R}^2;\mathbb{C}^2),
\end{equation}
where $\psi(t,\bx): \mathbb{R}\times \mathbb{R}^2 \to \mathbb{C}^2$ is a two-component spinor, $\hbar$ is the Planck constant, $c$ is the speed of light, $m$ is the rest mass of particle, $\boldsymbol{\alpha}$ and $\gamma^0$ are specific matrices associated with the Pauli matrices\footnote{see Appendix \ref{Appendix:Basics} for more details}, $\Id$ is the $2\times 2$ identity matrix, $V_\Gamma$ is an external potential with periodicity, and $V_{\mathrm{int}}$ is the interactive potential defined via the convolution, i.e., the Hatree-type interaction:
\begin{equation}\label{eq:Vint}
	V_{\mathrm{int}}(t,\bx) := (K*\rho)(t,\bx)=\int_{\R^2} K(\bx - \by)\, \rho(t,\by)\,\dd\by,
	\qquad \text{with} \quad \rho(t,\bx) := \overline{\psi(t,\bx)} \gamma^0 \psi(t,\bx),
\end{equation}
for a certain class of symmetric kernel functions $K$.

Inspired by the scalings in \cite{BMP2001, ELY2011}, we introduce $\ell$ as the microscopic lattice length, $L$ and $T = L/c$ as the macroscopic length and time scale, then by considering the following scalings in \eqref{eq:massive-dirac},

\begin{equation*}
	\tilde{\bx} = \frac{\bx}{L}, \quad \tilde{t} = \frac{t}{T}, \quad \ve = \frac{\ell}{L}, \quad \tilde{\hbar} = \frac{c \hbar}{L}, \quad V_\Gamma(\bx) = \tilde{V}_\Gamma\left( \frac{\bx}{\ell} \right), \quad V_{\mathrm{int}}(t,\bx) =V_{\mathrm{int}}\left(\frac{t}{T},\frac{\bx}{ L}\right),
\end{equation*}
we obtain, after dropping the tildes, the scaled Dirac equation in the \textit{massive} case,
\begin{equation}\label{eq:massive-dirac-scaled}
	\im \hbar \partial_t \psi^{\hbar}(t, \bx) = \Big( -\im \hbar c\, \boldsymbol{\alpha} \cdot \nabla + m c^2\,\gamma^0 + V_\Gamma\left(\frac{\bx}{\ve}\right)\Id  + V_{\mathrm{int}}(t,\bx) \Id \Big) \psi^{\hbar}(t, \bx),
\end{equation}
where both $\hbar$ and $\ve$ are small and their ratio determines the different asymptotic regimes, and for simplicity we omit the dependence notation on $\ve$ for $\psi^{\hbar}$, since we focus more on the semiclassical limit in this paper. For the \textit{massless} case, the massive equation \eqref{eq:massive-dirac-scaled} reduces to
\begin{equation}\label{eq:massless-dirac-scaled}
	\im \hbar \partial_t \psi^{\hbar}(t, \bx) = \Big( -\im \hbar c\, \boldsymbol{\alpha} \cdot \nabla + V_\Gamma\Big(\frac{\bx}{\ve}\Big)\Id  + V_{\mathrm{int}}(t,\bx) \Id \Big) \psi^{\hbar}(t, \bx).
\end{equation}

The aim of this paper is to rigorously derive the effective transport dynamics of the two-dimensional Dirac equation in the semiclassical limit $\hbar \to 0$, in the presence of both a periodic external potential and a Hartree-type interactive potential, for both the massive and massless cases.

The Wigner transform serves as the natural tool for this analysis, providing a phase-space formulation of quantum mechanics analogous to classical mechanics. However, the Dirac equation presents a significant challenge that is absent in the scalar Schr\"odinger equation. While the Wigner transform of the Schr\"odinger equation yields a scalar function that converges to the solution of the classical Liouville equation \cite{AM1976}, the inherent matrix structure of the Dirac system produces a matrix-valued Wigner function $\cW^\hbar$ (see \eqref{def:Wigner}), resulting in a complex coupled system that fails to provide a clear classical picture \cite{Spohn2000}.
To overcome this complexity and recover classical description, we perform a spectral decomposition of the Dirac Hamiltonian and project onto its positive and negative energy subspaces. This leads to the decoupled phase-space densities $f_\pm =\Tr[\Pi_\pm \cW^\hbar]$ as in \eqref{def:fpm}, which represents the semiclassical distribution of electrons in the conduction and valence bands, respectively. We show that, in the limit $\hbar \to 0$, these quantities satisfy distinct Vlasov or transport equations depending on the presence or absence of the mass term.

Briefly speaking, in the \textit{massive} case, the positive and negative energy bands are well-separated by a mass gap. We prove (Theorem \ref{thm:massiv}) that $f_\pm$ converge to solutions of the following relativistic Vlasov-type equation:
\begin{equation}\label{eq:relativisic-Vlasov}
	\partial_t f_\pm
	\pm \mathbf v(\bk)\cdot\nablax f_\pm
	- [\nablax V_\Gamma + \nablax K * \rho] \cdot\nablak f_\pm
	= 0,
\end{equation}
where
\begin{equation}\label{def:velocity}
	\bv(\bk):=\frac{c^2\,\bk}{\sqrt{(mc^2)^2+c^2|\bk|^2}}
	=\frac{c\,\bk}{\sqrt{m^2c^2+|\bk|^2}}
\end{equation}
and
\begin{equation}
	\rho(t,\bx)
	:= \int_{\R^2} \big( f_+(t,\bx,\bk) + f_-(t,\bx,\bk) \big)\,\dd\bk.
\end{equation}
Furthermore, in the \textit{massless} case, we prove (Theorem \ref{thm:massless}) that the band densities $f_\pm$ converge to the solution of a fixed-speed (nonlinear) transport equation:
\begin{equation}\label{eq:massless-transport}
	\partial_t f_\pm \pm c\,\widehat{\mathbf{k}}\cdot\nabla_{\bx} f_\pm
	- [\nablax V_\Gamma + \nablax K * \rho]      \cdot\nabla_{\mathbf{k}} f_\pm = 0,
\end{equation}
where $\widehat{\mathbf{k}}=\mathbf{k}/|\mathbf{k}|$ denotes the normalized momentum, and a truncation near the band-crossing region is imposed to exclude non-adiabatic coupling effects. The resulting limiting dynamics thus describe ballistic propagation at a fixed speed $c>0$, corresponding to the group velocity of the band energy.

\subsection{Background and motivation}

\noindent\textit{Background and Previous Results.}
The semiclassical limit of quantum mechanics to recover classical mechanics has been extensively investigated in the past decades.
For the Schr\"odinger dynamics, the convergence of Wigner transforms to Liouville or Vlasov type transport equations is classical \cite{AM1976}. In the absence of both external and interactive potentials, the semiclassical limit reduces to the Liouville equation. This result was rigorously established by G\'erard \cite{Gerard1991}, and subsequently confirmed by Lions--Paul \cite{LP1993} and Markowich--Mauser \cite{markowich1993classical}.
In the presence of a periodic potential, Markowich--Mauser--Poupaud \cite{MMP1994} introduced the concept of band-Wigner functions to study the semiclassical limit of electron motion in periodic media. The semiclassical limit of the Schr\"odinger--Poisson system was investigated by Brezzi--Markowich -- Castella \cite{BM1991, Castella1997}, and later extended to include periodic potentials by Bechouche \cite{Bechouche1999} and Bechouche--Mauser--Poupaud \cite{BMP2001} through the Wigner--Bloch expansion, under the assumption that the initial data is concentrated in isolated energy bands.
It is also worth mentioning that, for the pure states, Zhang--Zheng--Mauser \cite{ZZM2002} provided a rigorous justification of the semiclassical limit in one dimension for a very weak class of measure-valued solutions. Furthermore, G\'erard--Markowich--Mauser--Poupaud \cite{GMMP1997} investigated the homogenization limit, emphasizing the need for uniform asymptotic expansions when multiple limiting processes are involved.

In contrast to the Schr\"odinger equation, the semiclassical analysis of the Dirac system is far less developed. Early investigations, such as the WKB analysis by Pauli \cite{Pauli1932} and the block-diagonalization method by Foldy--Wouthuysen \cite{FoldyWouthuysen1950}, elucidated how spin enters the effective classical dynamics. Subsequent studies further explored asymptotic expansions \cite{RubinowKeller1963} and the Bargmann--Michel--Telegdi spin-precession law \cite{BMT1959}, providing the foundation for understanding spin transport in relativistic quantum mechanics.
In the 1980s, relativistic Wigner function approaches were developed through the introduction of covariant matrix-valued Wigner functions \cite{ElzeGyulassyVasak1986,ElzeGyulassyVasak1987}, offering a kinetic description consistent with Lorentz covariance. Later, the role of geometric phases and semiclassical corrections in multiband quantum transport was systematically studied from both physical and mathematical perspectives \cite{LittlejohnFlynn1991,SundaramNiu1999,PanatiSpohnTeufel2003,Vanderbilt2018}.
In multiband systems, the Wigner transform naturally becomes matrix-valued, encoding interband coherences. Under the assumption of a spectral gap, one obtains decoupling and effective classical dynamics within each energy band \cite{Spohn2000,Teufel2003}. This framework motivates the use of spectral projectors for Dirac systems, allowing the definition of band-resolved phase-space densities, which describe the semiclassical dynamics associated with the positive and negative energy subspaces.

Very recently, Golse--Leopold--Mauser--Möller--Saffirio \cite{GLMMS2025} established a global-in-time semiclassical limit from the Dirac equation in 3D with time-dependent external electromagnetic fields to relativistic Vlasov equations with Lorentz force. 
In contrast to the projected Wigner approach of \cite{GMMP1997}, their analysis takes the semiclassical limit directly at the level of the full matrix-valued Wigner equation and introduces a Lagrange multiplier to enforce the commutation constraint between the limiting Wigner measure and the Dirac symbol. 
This framework allows for lower regularity and time-dependent external fields, but does not address periodic backgrounds, band-crossing phenomena, or nonlinear Hartree-type interactions.

Furthermore, the influence of periodic external potentials becomes particularly significant in materials with honeycomb lattice symmetry, such as graphene which is a two-dimensional crystal composed of a single layer of carbon atoms \cite{2005NovoselovGeimMorozovJiangKatsnelsonGrigorievaDubonosFirsov}. The energy dispersion relation of graphene exhibits conical singularities, known as Dirac points, where the valence and conduction bands intersect \cite{FW2012JAMS, FW2014CMP, 2018BerkolaikoComech}. These cone-like structures, located at the vertices of the Brillouin zone, are responsible for the remarkable electronic, optical, and mechanical properties of graphene \cite{NGPNG2009}. More recently, twisted bilayer graphene — two stacked graphene sheets rotated by a so-called ``magic angle" — has attracted intense attention due to the emergence of strongly correlated electronic behavior and unconventional superconductivity \cite{Bistritzer2011, Cao2018}.
Such conical band crossings have a profound impact on the semiclassical limits of both Schr\"odinger and Dirac systems. For example, for massless 2D Dirac operators, conical intersections occur at points where the spectral projectors $\Pi_\pm$ become singular, and interband transitions play a crucial role. The rigorous analysis of wave propagation and transition probabilities near such crossings in the adiabatic and semiclassical regimes has been established in \cite{Hagedorn1991, LasserTeufel2005}. We also refer the reader to \cite{QiWangWatson2025} and the references therein for further discussion on the influence of band crossings in the semiclassical limit of the Schr\"odinger equation.

In addition, the nonlinear and mean-field models for Dirac systems provide a framework for quantifying the interplay between spectral structure, dispersion, and Coulomb interactions.
For instance, the relativistic quantum models based on the Dirac operator are central to the rigorous analysis of spectral and stability properties in many-body systems with Coulomb interactions, e.g. \cite{GLS999minimax, LSS1997stability}.
In the context of graphene, the Hartree-Fock theory captures non-perturbative ground-state properties \cite{HainzlLewinSparber2012}, and the global well-posedness of the corresponding time-dependent Hartree--Fock dynamics has been recently established by Hainzl--Lewin--Sparber \cite{BorrelliMorellini2025}. Moreover, small-data scattering results for two-dimensional Dirac--Hartree equations have been obtained by Cho--Lee--Ozawa \cite{ChoLeeOzawa2022}.
At the kinetic level, relativistic Vlasov equations emerge as semiclassical or mean-field limits of fermionic many-body systems and Hartree--Fock dynamics \cite{BPS2014rel, AkiMarSpa08, DietlerRademacherSchlein2018, LeoSaf23}. 
Further studies have addressed decay estimates and the non-relativistic limit of such models in three dimensions \cite{Wang2023, HongPankavich2025}.
It is also worth noting that, in the presence of random external perturbations, the semiclassical limit is expected to yield a radiative transport-type equation. This phenomenon was first rigorously identified by Spohn \cite{Spohn1977}, who derived such a limit for time-dependent Gaussian random impurities. Subsequent works have extended this framework and provided refined results on the impact of randomness on semiclassical propagation \cite{BFPR1999JSP, EY2000CPAM, BPR2002, QiWangWatson2025}. More recently, random Dirac dynamics have been investigated by Bal--Gu--Pinaud \cite{BGP2018}, highlighting the interplay between disorder and relativistic quantum transport.

\smallskip\smallskip
\noindent\textit{Motivation and Our Contributions.} 
Motivated by the physical intuitions and existing results above, we rigorously derive the effective transport dynamics of the 2D Dirac equation in the semiclassical limit, under the influence of a periodic external potential and a Hartree-type interactive potential, for both the massive and massless cases. 

The main difficulty arises from the matrix-valued nature of the Dirac Wigner transform, which mixes spinor components and prevents a direct passage to a scalar classical limit.
To overcome this challenge, we adopt a spectral projection approach, separating the dynamics into positive and negative energy components that correspond to electrons in the conduction and valence bands distributions. This decoupling allows us to identify the limiting transport equations governing the diagonal parts of the Wigner matrix and to distinguish between two qualitatively different regimes, meanwhile, we also have to rigorously justify the propagation of ``smallnes''s of off-diagonal part.
In the massive case, where a spectral gap ensures adiabatic separation, the limiting dynamics are shown to follow a relativistic Vlasov-type flow. In contrast, in the massless case, relevant for graphene and other Dirac materials with conical band crossings, therefore, the delicate cutoff need to be imposed to restrict the dynamics away from band-crossing regions to carefully handle the interband coupling, which further leads to the ballistic transport equation as the semiclassical limit. 

Additional analytical difficulties arise from controlling the error terms generated by the kinetic transport structure and by oscillatory potential interactions. In particular, the coupling between the matrix-valued Wigner dynamics and the external potential introduces delicate commutator structures that resist classical semiclassical expansions, while Hartree-type nonlocal terms induce self-consistent feedbacks that amplify oscillations at different scales. Establishing uniform bounds and precise error estimates in this setting therefore requires a subtle balance between dispersive analysis, microlocal structure, and mean-field regularity.

\subsection{Organization of the paper}
The rest of the paper is organized as follows.
In Section~\ref{sec:pre}, we introduce the Wigner transform together with its positive and negative energy projections, which constitute the main objects of this work.
Section~\ref{sec:main} presents the main results.
The proofs for the massive and massless cases are provided in the subsequent Section \ref{sec:massive}, addressing the estimates of the error terms.
In addition, some standard notations and useful preliminary lemmas are given in the Appendix.

\section{Preliminaries}\label{sec:pre}

\subsection{Notation} \label{subsec:notation}

Let $\{ \mathbf{e}_1,..., \mathbf{e}_d\}$ be a basis in $\R^d$, we define a periodic lattice $\Gamma$, and its dual lattice $\Gamma^*$ such that: 
\begin{equation*}
	\Gamma := \big\{ \sum_{j=1}^d m_j \mathbf{e}_j \, \big| \, m_j \in \mathbb{Z}, \, j=1,...,d \big\} 
	\quad\text{and}\quad
	\Gamma^* := \big\{ \sum_{\ell=1}^d m_l \mathbf{e}^\ell \, \big| \, m_\ell \in \mathbb{Z}, \, \ell=1,...,d \big\} \,,
\end{equation*}
where $\mathbf{e}^\ell$ is the dual basis of $\mathbf{e}_j$ in the sense that $(\mathbf{e}_j \cdot \mathbf{e}^\ell ) = 2 \pi \delta_{j\ell}.$
In addition, we define the \emph{fundamental cell} of $\Gamma$ as $\mathcal{C}$ and the \emph{Brillouin zone} as $\mathcal{B}$ such that:
\begin{equation}\label{def:fundamentalcell-Brillouinzone}
	\mathcal{C} := \big\{ \sum_{j=1}^d \theta_j \mathbf{e}_j \, \big| \, 0 \leq \theta_j < 1, \, j=1,...,d \big\}
	\quad\text{and}\quad
	\mathcal{B} := \big\{ \sum_{\ell = 1}^d \theta_\ell \mathbf{e}^\ell\, \big| \, 0 \leq \theta_\ell < 1, \, \ell = 1,...,d \big\}\,.
\end{equation}
Throughout this paper, the real-valued potential function $V_\Gamma(\bx)$ with $\Gamma$-periodicity refers to
\begin{equation*}
	V_\Gamma(\bx + \bmu) = V_\Gamma(\bx), \quad \forall \, \bx \in \R^d, \quad \bmu \in \Gamma. 
\end{equation*}

\subsection{Wigner transform}
We define the \textit{(symmetric) Wigner transform} as follows: for any pure state $\psi(t,\bx)$,
\begin{equation}\label{def:Wigner}
	\cW^{\hbar}(t,\bx,\mathbf{k})
	:=\frac{1}{(2\pi)^2}\!\int_{\mathbb{R}^2}\!
	\ee^{-\ii \mathbf{k} \cdot \mathbf{y}}\;
	\psi \!\left(t,\bx+\tfrac{\hbar}{2}\mathbf{y}\right)
	\overline{\psi\!\left(t,\bx-\tfrac{\hbar}{2}\mathbf{y}\right)}\gamma^{0}\;\mathrm{d}\mathbf{y},
\end{equation}
here, $\cW^{\hbar}$ is a $2 \times 2$ matrix-valued function on the phase space,
\begin{equation}\label{def:Wigner-matrix}
	\cW^{\hbar}(t, \bx, \mathbf{k}) =
	\begin{pmatrix}
		W_{11}^\hbar(t, \bx, \mathbf{k}) & -W_{12}^\hbar(t, \bx, \mathbf{k}) \\[4pt]
		W_{21}^\hbar(t, \bx, \mathbf{k}) & -W_{22}^\hbar(t, \bx, \mathbf{k})
	\end{pmatrix}
\end{equation}
with
\begin{align*}
	W_{ij}^\hbar(t, \bx, \mathbf{k}) &:= \frac{1}{(2\pi)^2}\!\int_{\mathbb{R}^2}\! \ee^{ - \ii \mathbf{k} \cdot \mathbf{y} }\, 
	\psi_i \left(t, \bx + \tfrac{\hbar}{2} \mathbf{y} \right)\, 
	\overline{\psi_j \left(t, \bx - \tfrac{\hbar}{2} \mathbf{y} \right)}\,\dd\mathbf{y} \quad \text{for} \quad i,j\in\{1,2\}.
\end{align*}

\begin{remark}
	The Wigner measure is not necessarily positive and may indeed take negative values. Nevertheless, the support of its negative part becomes negligible as $\hbar \to 0$. If it is important that the limiting object be a genuine probability measure, as is the case for solutions of the Vlasov equation, one may instead consider the Husimi measure, where some related results for Schr\"odinger system can be found in \cite{AFFGP11, CLL2021, FLP12, GP17}.
\end{remark}

\medskip

For future reference, we state a useful lemma describing the general relation between the Dirac charge density function and the trace of the corresponding Wigner transform.
\begin{lemma}[Spatial density as the $\bk$-marginal of the Wigner transform]
	\label{lem:rho-wigner-identity}
	Let $\psi^\hbar(t,\bx)\in L^2(\R^2;\C^2)$ and let $\cW^\hbar(t,\bx,\bk)\in \C^{2\times 2}$ be the associated matrix-valued Wigner transform
	\[
	\cW^{\hbar}(t,\bx,\bk)
	= \frac{1}{(2\pi)^2} \int_{\mathbb{R}^2}
	\ee^{-\ii \by \cdot \bk}\,
	\psi^\hbar\!\left(t, \bx + \tfrac{\hbar}{2} \by \right)
	\overline{\psi^\hbar \!\left(t, \bx - \tfrac{\hbar}{2} \by \right)}
	\gamma^0
	\,\dd \by,
	\]
	where $\overline{\psi} := \psi^\dagger \gamma^0$ is the Dirac adjoint.
	Then, the Dirac density function satisfies
	\begin{equation}\label{eq:rho-Wigner-marginal}
		\rho^\hbar(t,\bx)
		:= \bar\psi^\hbar(t,\bx)\,\gamma^0\,\psi^\hbar(t,\bx)
		\;=\;
		\int_{\R^2} \Tr [\cW^\hbar](t,\bx,\bk)\,\dd\bk.
	\end{equation}
\end{lemma}

\begin{proof}
	A direct use of cyclicity of the trace gives
	\[
	\Tr\!\left[
	\psi^\hbar\!\left(\bx + \tfrac{\hbar}{2}\by\right)
	\overline{\psi^\hbar\!\left(\bx - \tfrac{\hbar}{2}\by\right)}
	\gamma^0 \right]
	= \overline{\psi^\hbar\!\left(\bx - \tfrac{\hbar}{2}\by\right)}\;
	\gamma^0\, \psi^\hbar\!\left(\bx + \tfrac{\hbar}{2}\by\right).
	\]
	Hence
	\[
	\Tr [\cW^\hbar](t,\bx,\bk)
	= \frac{1}{(2\pi)^2}\!
	\int_{\R^2}
	\ee^{-\ii\by\cdot\bk}\,
	\overline{\psi^\hbar\!\left(\bx - \tfrac{\hbar}{2}\by\right)}\;
	\gamma^0\,
	\psi^\hbar\!\left(\bx + \tfrac{\hbar}{2}\by\right)
	\dd\by.
	\]
	Integrating over $\bk$ and using
	\[
	\frac{1}{(2\pi)^2}\int_{\R^2} \ee^{-\ii\by\cdot\bk}\dd\bk = \delta(\by)
	\]
	yields
	\[
	\int_{\R^2} \Tr[\cW^\hbar](t,\bx,\bk)\,\dd\bk
	= \overline{\psi^\hbar(t,\bx)}\;\gamma^0\,\psi^\hbar(t,\bx),
	\]
	which is exactly \eqref{eq:rho-Wigner-marginal}.
\end{proof}

\medskip

\begin{remark}[Band decomposition of the spacial density]
	In our $2\times2$ Dirac representation $\gamma^0=\sigma^3$, the density reduces to
	\[
	\rho^\hbar(t,\bx)
	= |\psi_1^\hbar(t,\bx)|^2 + |\psi_2^\hbar(t,\bx)|^2
	= \langle \psi^\hbar(t,\bx),\psi^\hbar(t,\bx)\rangle_{\C^{2}}.
	\]
	Let $\Pi_\pm(\bk)$ be the spectral projectors of the Dirac symbol, and define the band Wigner functions
	\[
	f_\pm^\hbar(t,\bx,\bk)
	:= \Tr\!\big[\Pi_\pm(\bk)\,\cW^\hbar(t,\bx,\bk)\big].
	\]
	Since $\Pi_+(\bk)+\Pi_-(\bk)=\Id$, we have
	\[
	\Tr[\cW^\hbar](t,\bx,\bk)
	= f_+^\hbar(t,\bx,\bk) + f_-^\hbar(t,\bx,\bk),
	\]
	and by Lemma~\ref{lem:rho-wigner-identity},
	\begin{equation}\label{eq:rho-band}
		\rho^\hbar(t,\bx)
		= \int_{\R^2} \big(f_+^\hbar(t,\bx,\bk)
		+ f_-^\hbar(t,\bx,\bk)\big)\,\dd\bk.
	\end{equation}
	Therefore, the Hartree-type potential $V_{\mathrm{int}}$ in the 2D Dirac--Hartree model couples to the total band density, with both bands contributing with the same sign.
\end{remark}

\medskip

\begin{remark}[Physical interpretation of the two bands] \label{rem:bands-interpretation}
	The Dirac symbol $H_m(\bk)=c\,\balpha\!\cdot\!\bk + mc^2\gamma^0$ has two eigenvalues $E_\pm(\bk)=\pm E_m(\bk)$, which correspond to two Bloch bands (conduction and valence). In the present $2\times2$ effective Dirac model, both bands describe electronic states and therefore carry the same physical charge. In particular, the spatial density entering the Hartree potential is
	\[
	\rho^\hbar(t,\bx)
	= \int_{\R^2} \Tr[\cW^\hbar](t,\bx,\bk)\,\dd\bk
	= \int_{\R^2} \big(f_+^\hbar(t,\bx,\bk)
	+ f_-^\hbar(t,\bx,\bk)\big)\,\dd\bk,
	\]
	so that each band contributes with the same sign.
	
	The opposite signs in the transport part $\partial_t f_\pm  \pm \bv(\bk) \cdot \nabla_\bx f_\pm$ of the limiting Vlasov equations
	
	do not reflect opposite charges. They arise solely from the fact that the group velocities of the two bands satisfy $\nabla_\bk E_+(\bk) = -\,\nabla_\bk E_-(\bk)$.
	In other words, electrons in the conduction and valence bands drift in the opposite directions due to the opposite slopes of the two dispersion relations, while the force term $-\nabla_\bx V\cdot\nabla_\bk f_\pm$ remains the same, because their physical charge is identical.
	
	We emphasize that, unlike in the relativistic $4\times4$ Dirac theory, the negative-energy band here does not represent antiparticles, and no charge-sign difference appears.
\end{remark}

\subsection{Positive and negative energy projection}
\label{sec:proj-just}

We introduce the positive and negative energy projected operator $\Pi_{\pm}(\bk)$.

By the spectral theorem, for any self-adjoint operator $A$, we can define $\mathbf{1}_{(0,\infty)}$ as the spectral projector onto its positive spectrum in the sense that
\begin{equation}\label{def:pro}
	\mathbf{1}_{(0,\infty)}(A)=\frac12\left(\mathbf{1}+ \operatorname{sgn}(A)\right),
	\quad \text{for} \quad
	\operatorname{sgn}(A):=A\,|A|^{-1}.
\end{equation}

We denote $H_m(\bk)$ as the massive Dirac Bloch symbol in 2D:
\begin{equation}\label{def:Hm}
	H_m(\bk):=c\,\boldsymbol{\alpha}\!\cdot\!\bk + m c^2 \gamma^0
	= c\,(k_1\sigma_1+k_2\sigma_2)+ m c^2 \gamma^0,
\end{equation}
where
\[ 
H_m(\bk)=H_m^\ast(\bk) \quad  \text{and} \quad H_m^2(\bk) = \big(c^2|\bk|^2+(mc^2)^2\big) \, \Id,
\]
and the corresponding eigenvalues of $H_m(\bk)$ are 
\begin{equation}\label{def:Em}
	\pm E_m(\bk):= \pm \sqrt{c^2|\bk|^2+(mc^2)^2}.
\end{equation}

Then, substituting $A=H_m(\bk)$ into \eqref{def:pro}, we obtain the \emph{positive/negative energy projector} (Riesz projectors onto the positive/negative spectral branch) of $H_m(\bk)$ that
\begin{equation}\label{eq:proj-general}
	\begin{aligned}
		\Pi_{+}(\bk)
		&:=\mathbf{1}_{(0,+\infty)}\!\bigl(H_m(\bk)\bigr)
		= \frac12\!\Big(\Id + \frac{H_m(\bk)}{E_m(\bk)}\Big)
		=\frac12\!\Big(\Id + \frac{c\,\balpha\!\cdot\!\bk + m c^2 \gamma^0}{E_m(\bk)}\Big),\\
		\Pi_{-}(\bk)
		&:=\mathbf{1}_{(-\infty,0)}\!\bigl(H_m(\bk)\bigr)
		= \frac12\!\Big(\Id - \frac{H_m(\bk)}{E_m(\bk)}\Big)
		=\frac12\!\Big(\Id - \frac{c\,\balpha\!\cdot\!\bk + m c^2 \gamma^0}{E_m(\bk)}\Big),
	\end{aligned}
\end{equation}
where we use $|H_m(\bk)|=\sqrt{H_m^2(\bk)}=E_m(\bk)\, \Id$ in the second equality above.

Note that, for the massless case ($m=0$),
\begin{equation}\label{eq:proj-massless}
	\Pi_{\pm}(\bk)
	:=\frac12\!\big(\Id \pm \hat\bk\!\cdot\!\boldsymbol{\alpha}\big),
\end{equation}
where we denote $\hat\bk:=\bk/|\bk|$ for $\bk \neq \mathbf{0}$.

\medskip

One can verify that the projectors $\Pi_{\pm}(\bk)$ are smooth on $\R^2_{\bk}$ and satisfy that
\begin{equation}\label{eq:proj-basic}
	\begin{aligned}
		&\Pi_\pm^2 (\bk)=\Pi_\pm (\bk), 
		\qquad
		\Pi_+(\bk) \; \Pi_-(\bk) = 0, 
		\qquad
		\Pi_+ (\bk) + \Pi_-(\bk) =\Id,\\[5pt]
		&[H_m(\bk),\Pi_\pm(\bk)]=0, \quad
		H_m(\bk)\;\Pi_\pm(\bk)=\pm E_m(\bk)\;\Pi_\pm(\bk),\\[5pt]
		&\Pi_\pm\,(\nablak \Pi_\pm)\,\Pi_\pm=0,
		\quad
		\Pi_\pm\,(\nablak \Pi_\mp)\,\Pi_\pm=0,
	\end{aligned}
\end{equation}
where the proof of the last two identities is provided in the Appendix \ref{Appendix:Basics}.

Thanks to the projectors $\Pi_\pm$, it is useful to decompose the Wigner matrix \eqref{def:Wigner-matrix} into its diagonal and off-diagonal parts:
\begin{equation}\label{def:WD+WOD}
	\cW^{\hbar}(t,\bx,\bk) := \cW^{\hbar}_{\mathrm{D}}(t,\bx,\bk) + \cW^{\hbar}_{\mathrm{OD}}(t,\bx,\bk),
\end{equation}
with
\begin{equation}
	\cW^{\hbar}_{\mathrm{D}}(t,\bx,\bk) := \Pi_+\cW^{\hbar}\Pi_+ + \Pi_-\cW^{\hbar}\Pi_-,
	\qquad
	\cW^{\hbar}_{\mathrm{OD}}(t,\bx,\bk) := \Pi_+\cW^{\hbar}\Pi_- + \Pi_-\cW^{\hbar}\Pi_+.
\end{equation}

\begin{lemma}[Derivatives of the spectral projectors]\label{lem:grad-Pi}
	The following bounds hold:  
	\begin{align}
		&\|\nabla_\bk \Pi_\pm^m(\bk)\|_{\mathrm{op}}
		\le \frac{1}{2mc}, && \text{for massive case }(m>0), \label{eq:gradPi-massive}\\[3pt]
		&\|\nabla_\bk \Pi_\pm^0(\bk)\|_{\mathrm{op}}
		\le \frac{C}{|\bk|}, \quad \bk\neq 0, && \text{for massless case }(m=0), \label{eq:gradPi-massless}
	\end{align}
	where $\|\cdot\|_{\mathrm{op}}$ is the usual operator norm.
\end{lemma}

\begin{remark}\label{rem:k-in-B}
	In the case of periodic external potential as in \eqref{def:fundamentalcell-Brillouinzone}, the quasi-momentum $\bk$ belongs to the Brillouin zone $\mathcal{B}$.
	All the definitions and identities of the projectors $\Pi_\pm(\bk)$ above hold pointwise for $\bk \in \mathcal{B}$, and their smoothness extends periodically to the boundary of $\mathcal{B}$.
	This observation will be relevant when we later integrate quantities involving $\Pi_\pm(\bk)$ over the Brillouin zone.
\end{remark}

\begin{proof}[Proof of Lemma \ref{lem:grad-Pi}]
	For \emph{massive} case:  
	when $m>0$,
	\[
	\Pi_\pm^m(\bk)=\frac{1}{2}\Big(\Id \pm\frac{H_m(\bk)}{E_m(\bk)}\Big),
	\qquad
	\nabla_\bk E_m(\bk)=\frac{c^2\bk}{E_m(\bk)}.
	\]
	Taking the gradient of $\Pi_\pm$ gives
	\begin{align*}
		\nabla_\bk \Pi_\pm^m(\bk)
		&=\pm\frac{1}{2}\Big(
		\frac{c\,\boldsymbol{\alpha}}{E_m(\bk)}
		-\frac{H_m(\bk)\,\nabla_\bk E_m(\bk)}{E_m(\bk)^2}
		\Big)\\[3pt]
		&=\pm\frac{c^3}{2E_m(\bk)^3}
		\Big(E_m(\bk)^2\,\boldsymbol{\alpha}
		-c^2(\boldsymbol{\alpha}\!\cdot\!\bk)\,\bk
		-m c^2 E_m(\bk)\,\gamma^0\,\bk\Big),
	\end{align*}
	where each term is smooth and bounded on $\R^2$, hence
	\[
	\|\nabla_\bk \Pi_\pm^m(\bk)\|_{\mathrm{op}}
	\le \frac{c}{2E_m(\bk)}
	\le \frac{1}{2mc},
	\]
	leading to \eqref{eq:gradPi-massive}.
	The mass term $m c^2$ opens a spectral gap of size $2mc^2$, regularizing the projectors.
	
	For \emph{massless} case:  
	when $m=0$,
	\[
	H_0(\bk)=c\,\boldsymbol{\alpha}\!\cdot\!\bk,
	\qquad
	E_0(\bk)=c|\bk|,
	\qquad
	\Pi_\pm^0(\bk)
	=\frac{1}{2}\Big(\Id\pm\frac{H_0(\bk)}{E_0(\bk)}\Big)
	=\frac{1}{2}\big(\Id\pm\boldsymbol{\alpha}\!\cdot\!\hat{\bk}\big),
	\]
	with $\hat{\bk}=\bk/|\bk|$. Hence
	\[
	\nabla_\bk \Pi_\pm^0(\bk)
	=\frac{1}{2}\,\nabla_\bk(\boldsymbol{\alpha}\!\cdot\!\hat{\bk})
	=\frac{1}{2|\bk|}
	\big(\boldsymbol{\alpha}-(\boldsymbol{\alpha}\!\cdot\!\hat{\bk})\,\hat{\bk}\big),
	\]
	so that $\|\nabla_\bk \Pi_\pm^0(\bk)\|_{\mathrm{op}}\lesssim |\bk|^{-1}$.  
	The singularity at $\bk=\bk^*$ reflects the conical band crossing of the massless Dirac operator, motivating the introduction of a cutoff region
	$\Omega_\kappa=\{|\bk - \bk^*|\ge\kappa\}$ in the analysis.  
	In the massless case, the factor $1/(mc)$ is thus replaced by the singular weight $1/|\bk|$.
\end{proof}

\section{Main results} 
\label{sec:main}

\subsection{Assumptions}

Before presenting the main results, some assumptions are illustrated for the external potential $V_\Gamma$, interactive kernel $K$, initial condition $\psi_0$ and the corresponding initial Wigner measure $\cW^{\hbar}_0$,
\begin{equation}\label{eq:initial-W}
	\cW^{\hbar}_0(\bx,\mathbf{k}) = \cW^{\hbar}(t=0,\bx,\mathbf{k})
	:=\frac{1}{(2\pi)^2}\!\int_{\mathbb{R}^2}\!
	\ee^{-\ii \mathbf{k} \cdot \mathbf{y}}\;
	\psi_0\!\left(\bx+\tfrac{\hbar}{2}\mathbf{y}\right)
	\overline{\psi_0\!\left(\bx-\tfrac{\hbar}{2}\mathbf{y}\right)}\gamma^{0}\;\mathrm{d}\mathbf{y},
\end{equation}
which are critical in the proof of the main Theorems.

\begin{enumerate}
	\item[\bf (A1)]\textbf{Regularity of the potential and kernel.}
	The external potential $V_\Gamma \in A^2(\Gamma)$ is $\Gamma$–periodic, i.e.,
	\[
	\| \bmu^2 \hat{V_\Gamma}(\bmu) \|_{\ell^1_{\bmu}} :=\sum_{\bmu \in \Gamma^*} |\bmu|^2\,|\widehat V_\Gamma(\bmu)| < \infty,
	\]
	and the interactive kernel $K \in A^1(\R^2)$ in the Hartree-type interactive potential satisfies
	\[
	\| \bk' \hat{K}(\bk')\|_{L^1_{\bk'}} := \int_{\R^2} |\bk'| |\hat{K}(\bk')| \,\dd \bk' < \infty.
	\]
	
	\item[\bf (A2)]\textbf{Boundedness of the initial condition.}
    For the initial condition $\psi_0$ in \eqref{eq:initial-psi},
	\[
	\|\psi_0\|_{H^1_{\bx}(\mathbb{R}^2)}\le C_0.
	\]
	
	\item[\bf (A3)]\textbf{Small inter-band coherence of the initial condition.}
    The off-diagonal part of the initial Wigner measure $\cW^{\hbar}_0$ in \eqref{eq:initial-W} is small: for $\bx \in \mathbb{R}^2$,
	\[
	\big\|\mathbf{1}_{|\bk-\bk_*|>\kappa} \,\cW^\hbar_{0,\mathrm{OD}}(\bx, \cdot) \big\|_{H^{-1}_{\bk}(\mathbb{R}^2)} \, \le \, \bar{C}_0 \,\hbar,
	\]
	where $\cW^{\hbar}_{0,\mathrm{OD}}:=\Pi_+ \cW^{\hbar}_0 \Pi_- + \Pi_- \cW^{\hbar}_0 \Pi_+ $.
	
	\item[\bf (A4)]\textbf{Support of the initial condition away from the band crossings.}
    The initial measure $\cW^{\hbar}_0$ in \eqref{eq:initial-W} is supported $\kappa$-away from the band crossings: for $\bx \in \mathbb{R}^2$,
	\[
	\mathbf{1}_{|\bk-\bk_*|>\kappa}\cW^{\hbar}_0(\bx,\bk)=\cW^{\hbar}_0(\bx,\bk).
	\]
\end{enumerate}

\medskip

\begin{remarks*}\phantom{ }
	\begin{enumerate}
		\item The regularity required in Assumption~\textnormal{\textbf{(A1)}} is mainly due to the interactions between the energy bands that possibly happen.
        In fact, the estimates for the kinetic related error term only rely on the boundedness of the potentials and their first derivatives,
        while the estimates of the error terms need decay property in the Fourier space (e.g., \eqref{est:Q-ext-error} and \eqref{est:Q-int-error}), where the Wiener-algebra hypotheses $V_\Gamma \in A^{2}(\Gamma)$ and $K \in A^{1}(\R^{2})$ are sufficient.
        Thus, \textnormal{\textbf{(A1)}} provides a unified and minimal regularity requirement for both the kinetic and potential parts. For more details about the weighted Wiener algebras and their symbolic properties, we refer to~\cite{BMP2001}.
		
		\item Assumption \textnormal{\textbf{(A2)}} is imposed on the initial data. By the $L^2_\bx$–conservation law (Lemma~\ref{lem:L2conservation}), this property persists uniformly for all $t>0$ and $\hbar\in(0,1]$. Moreover, the $H^1_\bx$-norm remains uniformly bounded; see Lemma~\ref{lem:bdd-H1}.
		
		\item  Assumption \textnormal{\textbf{(A3)}} describes that (i) The state is prepared almost entirely within one energy band (diagonal in the band basis); (ii) inter-band coherence (off-diagonal term) is suppressed by a factor of $\hbar$; (iii) this matches the prediction of the adiabatic theorem: transitions between bands are of order $\hbar$ under slow/semiclassical evolution.
		
		It is the standard \emph{adiabatic ansatz} of the initial state (leading symbol diagonal in the band basis). It is necessary in the massless case, where the band crossing points may appear. 
		
		\item Assumption \textnormal{\textbf{(A4)}} requires that the initial Wigner distribution  $\cW^{\hbar}_0$ in \eqref{eq:initial-W} vanishes in a neighborhood of the Dirac point $\bk=\bk_*$. Equivalently, the Fourier transform of the initial spinor $\widehat\psi_0(\bk)$ has support bounded away from the band crossings, i.e., $|\bk - \bk_*|\ge\kappa$. Typical examples include semiclassical wave packets
		\[
		\psi^\hbar(\bx) = \hbar^{-1} a\!\left(\tfrac{\bx-\bx_*}{\hbar}\right)\, 
		\ee^{\ii \frac{(\bk-\bk_*) \cdot \bx}{\hbar} }, \qquad |\bk-\bk_*| \ge \kappa,
		\]
		whose Wigner transform is away from  $\bk_*$, or more generally, the band-adapted superpositions are supported in $\{ \bk \,| \,|\bk-\bk_*| \ge \kappa\,\}$. 
		Thus, \textnormal{\textbf{(A4)}} means the initial state carries no contribution from the region 
		around the band crossings $\bk_*$, where the band projectors are singular.
		\begin{figure}[h!]
			\centering
			\includegraphics[width=0.65\textwidth]{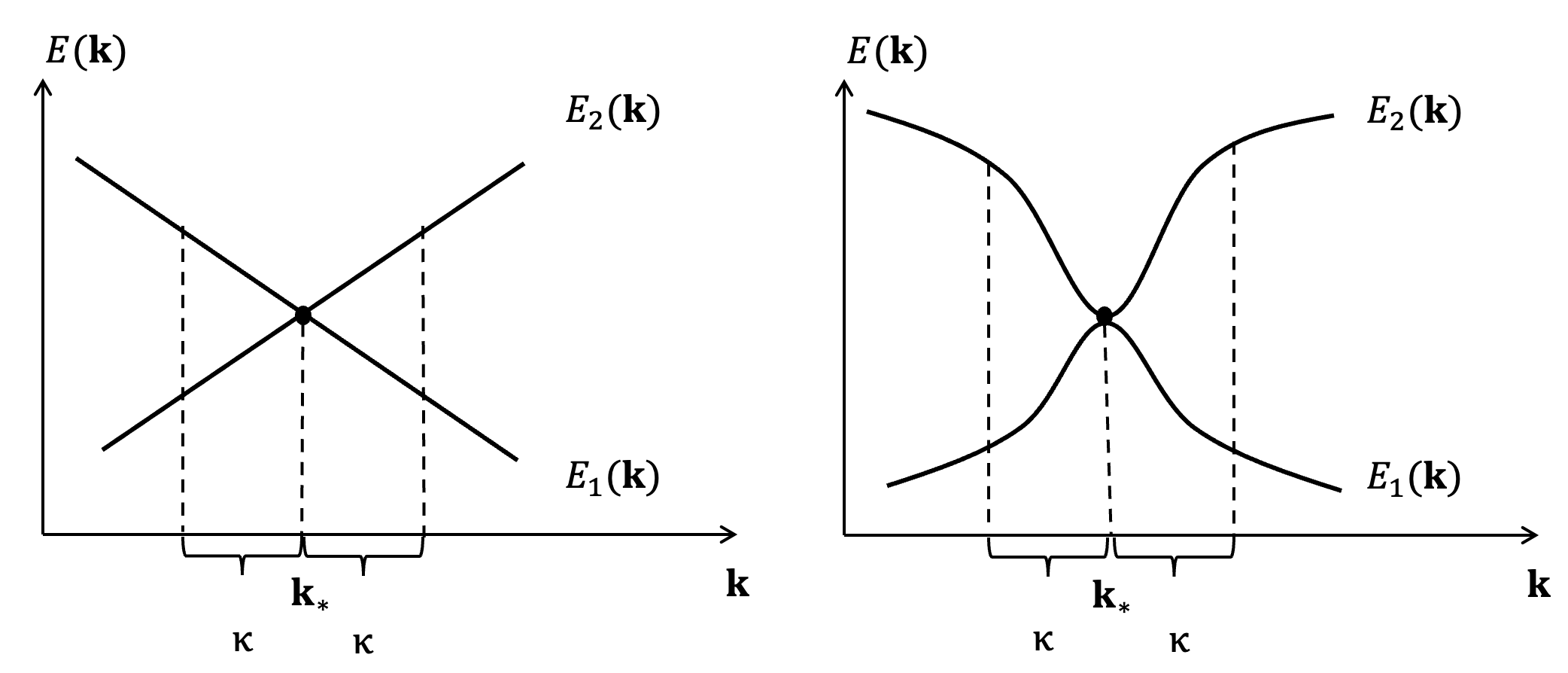}
			\caption{Left: A conical energy bands crossing with $\kappa$-neighborhood. Right: A quadratic energy bands crossing with $\kappa$-neighborhood \cite{QiWangWatson2025}.}
			\label{fig1}
		\end{figure}
	\end{enumerate}
\end{remarks*}

\subsection{Main Theorems}

Now we are in a position to present our main results. We first clarify the basic setup, i.e., in which sense we would talk about the convergence.

We study the Dirac system in the limit $\hbar \to 0$ through the evolution of the Wigner transform \eqref{def:Wigner}.
By applying the Wigner transform to both sides of \eqref{eq:massive-dirac-scaled}, we can formally obtain the corresponding Wigner equation\footnote{See Appendix \ref{Appendix:Componentwise} for complete derivation.}, 
\begin{multline}\label{eq:Wigner}
	\partial_t \cW^{\hbar} (t,\bx,\bk)
	+ c\,\boldsymbol{\alpha}\!\cdot\!\nabla_{\bx} \cW^{\hbar}(t,\bx,\bk)
	+ \frac{1}{\ii\hbar}\,[H_m,\,\cW^{\hbar}](t,\bx,\bk)\\ 
	- \mathcal Q^\ve[V_\Gamma, \cW^{\hbar}](t,\bx,\bk) - \mathcal Q^{\hbar}[V_{\mathrm{int}}, \cW^{\hbar}](t,\bx,\bk)=\mathbf{0},
\end{multline}
with the initial condition \eqref{eq:initial-W}.

Here $H_m:=H_m(\bk)$ is the massive Dirac Bloch symbol in 2D defined in \eqref{def:Hm},
the external potential-related term  $\mathcal Q^\ve[V_\Gamma, \cW^{\hbar}]$ reads,
\begin{equation}\label{eq:Q-L}
	\begin{aligned}
		\mathcal Q^\ve[V_\Gamma, \cW^{\hbar}](t,\bx,\bk)
		&:= \frac{\ii}{\hbar} \int_{\R^2}\frac{\dd\by}{(2\pi)^2} \ee^{-\ii \bk \cdot \by} \left[ V_\Gamma \left( \frac{\bx}{\ve} + \frac{\hbar}{2\ve} \by \right) - V_\Gamma \left( \frac{\bx}{\ve} - \frac{\hbar}{2\ve} \by \right) \right] \psi^{\hbar} \left(t,\bx+\tfrac{\hbar}{2}\mathbf{y}\right)
		\overline{\psi^{\hbar} \left(t,\bx-\tfrac{\hbar}{2}\mathbf{y}\right)}\gamma^{0} \\[5pt]
		&= \frac{\ii}{\hbar} \int_{\R^2}\frac{\dd\by}{(2\pi)^2}
		\int_{\R^2}\!\dd\bk'\,
		\ee^{\ii (\bk - \bk') \cdot \by} \left[ V_\Gamma \left( \frac{\bx}{\ve} + \frac{\hbar}{2\ve} \by \right) - V_\Gamma \left( \frac{\bx}{\ve} - \frac{\hbar}{2\ve} \by \right) \right]
		\cW^{\hbar}(t,\bx,\bk')\\[5pt]
		&= \frac{\ii}{\hbar} \sum_{\bmu \in \Gamma^*} \ee^{\ii \bmu \cdot \frac{\bx}{\ve}} \hat{V}_{\Gamma}(\bmu) \left[ \cW^{\hbar} \left(t,\bx, \bk +\frac{\hbar}{2\ve}\bmu \right) - \cW^{\hbar} \left(t,\bx, \bk -\frac{\hbar}{2\ve}\bmu \right) \right],
	\end{aligned}
\end{equation}
where 
\[
\hat{V}_{\Gamma}(\bmu) = \frac{1}{|\mathcal{C}|} \int_{\mathcal{C}} \ee^{-\ii \bmu\cdot \by} V_\Gamma(\by) \, \dd \by, 
\]
and the periodicity of the external potential $V_\Gamma$ is considered in the last equality above;
also the interactive Hartree-related term $\mathcal Q^{\hbar}[V_{\mathrm{int}}, \cW^{\hbar}]$ corresponds to
\begin{equation}\label{eq:Q-int}
	\begin{aligned}
		\mathcal Q^{\hbar}[V_{\mathrm{int}},\cW^{\hbar}](t,\bx,\bk)
		&:= \frac{\ii}{\hbar} \int_{\R^2}\frac{\dd\by}{(2\pi)^2}
		\int_{\R^2}\!\dd\bk'\,
		\ee^{\ii (\bk - \bk') \cdot \by} \Big[ V_{\mathrm{int}} \left(t, \bx + \frac{\hbar \by}{2} \right) - V_{\mathrm{int}} \left(t, \bx - \frac{\hbar \by}{2} \right) \Big]
		\cW^{\hbar}(t,\bx,\bk')\\
		&= \frac{\ii}{\hbar} \int_{\R^2} \dd\by
		\int_{\R^2}\!\dd\bk'\,
		\ee^{\ii \bk' \cdot (\bx-\by)} V_{\mathrm{int}} \left(t, \by \right) 
		\Big[ \cW^{\hbar}\left(t,\bx,\bk+ \frac{\hbar\bk'}{2}\right) - \cW^{\hbar}\left(t,\bx,\bk- \frac{\hbar\bk'}{2}\right) \Big].
	\end{aligned}
\end{equation}

If further defining $f^{\hbar}_\pm$ as the \emph{positive/negative-energy-projected densities} of \eqref{eq:Wigner} in the sense that
\begin{equation}\label{def:fpm}
	f^{\hbar}_\pm(t,\bx,\bk):=\Tr[\Pi_\pm \cW^{\hbar}](t,\bx,\bk),
\end{equation}
then, in the following Theorem \ref{thm:massiv}, $f^{\hbar}_\pm$ is proved to converge to $f_\pm$, which is the solution to the limiting relativistic Vlasov-type equation, describing the semiclassical motion of the electrons in conduction and valence, respectively.

\begin{theorem}[Massive case ($m>0$)]\label{thm:massiv}
	For the potential terms $V_\Gamma, K$ satisfying \textnormal{\textbf{(A1)}}, let $\psi^{\hbar}$ be the solution to the scaled Dirac equation \eqref{eq:massive-dirac-scaled} with initial condition $\psi_0$ satisfying \textnormal{\textbf{(A2)}}, and its Wigner transform $\cW^{\hbar}$ be the solution to \eqref{eq:Wigner} with initial condition $\cW^{\hbar}_0$ satisfying \textnormal{\textbf{(A3)}}.
	
	Then, the positive/negative energy-projected densities $f^{\hbar}_\pm$ defined in \eqref{def:fpm} satisfy, in the distributional sense,
	\begin{equation}\label{eq:massive-eff-transport}
		\left\{
		\begin{aligned}
			\partial_t f^{\hbar}_\pm (t,\bx,\bk)
			\pm \mathbf v(\bk)\cdot\nablax f^{\hbar}_\pm (t,\bx,\bk)
			- \Big[ \nablax  V_\Gamma +\nablax K * \rho^\hbar \Big] \cdot \nablak f^{\hbar}_\pm(t,\bx,\bk) =\;& \Err_\pm^{\hbar}(t,\bx,\bk),\\[5pt]
			f^{\hbar}_\pm(t=0,\bx,\bk) =& \Tr[\Pi_\pm \cW^{\hbar}_0](\bx,\bk),
		\end{aligned}
		\right.
	\end{equation}
	where $\mathbf v(\bk)$ is the massive group velocity defined as $\mathbf v(\bk)=\tfrac{c^2}{E_m(\bk)}\,\bk$ and
	\[
	\rho^\hbar(t,\bx)
	:= \int_{\R^2} \big( f^\hbar_+(t,\bx,\bk) + f^\hbar_-(t,\bx,\bk) \big)\,\dd\bk,
	\]
	furthermore, the total error term $\Err_\pm^{\hbar}$ satisfies, for $\eta (\bx) \in W^{1,\infty}_\bx(\R^2)$ and $ \phi (\bk) \in W^{2,\infty}_\bk(\R^2)$, 
	\begin{equation}\label{est-total-massive}
		\sup_{t\in[0,T]}\ \sup_{\substack{ \|\eta\|_{W^{1,\infty}_\bx} \le 1\\ \|\phi\|_{W^{2,\infty}_\bk} \le 1}}
		\big|\langle \Err_\pm^{\hbar}(t,\bx,\bk), \eta(\bx) \phi(\bk)\rangle \big|
		\ \le\ C_T\, \hbar \left( 
		1 + \ve^{-2}\right).
	\end{equation}
	
	Consequently, as $\hbar \to 0$, $f^{\hbar}_\pm$ converges to $f_\pm$ in
	$C \big([0,T]; (W^{1,\infty}_\bx W^{2,\infty}_\bk(\Omega_\kappa))'\big)$.
	The limit functions $f_\pm$ are the solutions to the Vlasov equation
	\begin{equation}\label{eq:massive-limit}
		\left\{
		\begin{aligned}
			\partial_t f_\pm(t,\bx,\bk) \pm \mathbf v(\bk) \cdot \nablax f_\pm(t,\bx,\bk) - \nablax V \cdot \nablak f_\pm(t,\bx,\bk) &= 0,\\[4pt]
			f_\pm(t=0,\bx,\bk) &= \Tr[\Pi_\pm \cW_0](\bx,\bk),
		\end{aligned}
		\right.
	\end{equation}
	where
	\[
	\Tr\big[\Pi_\pm(\bk)\,\cW_0(\bx,\bk)\big]
	= \lim_{\hbar \to 0} \Tr\big[\Pi_\pm(\bk)\,\cW_0^\hbar(\bx,\bk)\big],
	\]
	and the mean-field potential is given by
	\begin{equation*}
		V(t,\bx)
		= V_\Gamma\!\left(\frac{\bx}{\ve}\right)
		+ (K * \rho)(t,\bx),
		\qquad
		\rho(t,\bx)
		:= \lim_{\hbar\to0}\int_{\R^2} \big( f^\hbar_+(t,\bx,\bk) + f^\hbar_-(t,\bx,\bk) \big)\,\dd\bk.
	\end{equation*}
\end{theorem}

Based on the massive case, we can present the similar result for the massless case ($m=0$), where the only difference lies in the kinetic term. 
In fact, the corresponding massless Wigner equation reads: for the initial condition \eqref{eq:initial-W},
\begin{multline}\label{eq:Wigner-massless}
	\partial_t \cW^{\hbar} (t,\bx,\bk)
	+ c\,\boldsymbol{\alpha}\!\cdot\!\nabla_{\bx} \cW^{\hbar}(t,\bx,\bk)
	+ \frac{1}{\ii\hbar}\,[H_0,\,\cW^{\hbar}](t,\bx,\bk)\\ - \mathcal Q^\ve[V_\Gamma, \cW^{\hbar}](t,\bx,\bk) - \mathcal Q^{\hbar}[V_{\mathrm{int}}, \cW^{\hbar}](t,\bx,\bk)=\mathbf{0},
\end{multline}
where the massless Dirac Bloch symbol $H_0$ is $H_0:=c\,\boldsymbol{\alpha}\!\cdot\!\bk$

In the following Theorem \ref{thm:massless}, by taking the projection \eqref{def:fpm}, the energy-projected densities $f^{\hbar}_\pm$ are shown to satisfy a similar effective equation in contrast to the massive case, nevertheless, the group velocity $\mathbf{v}(\bk)$ degenerates. 

\begin{theorem}[Massless case ($m=0$)] \label{thm:massless}
	For the potential terms $V_\Gamma, K$ satisfying \textnormal{\textbf{(A1)}}, let $\psi^{\hbar}$ be the solution to the scaled Dirac equation \eqref{eq:massless-dirac-scaled} with initial condition $\psi_0$ satisfying \textnormal{\textbf{(A2)}}, and its Wigner transform $\cW^{\hbar}$ be the solution to \eqref{eq:Wigner-massless} with the initial condition $\cW^{\hbar}_0$ satisfying \textnormal{\textbf{(A3)}} and \textnormal{\textbf{(A4)}}.
	
	Then, the energy-projected densities $f^{\hbar}_\pm$ satisfy, in the distributional sense,
	\begin{equation}\label{eq:massless-eff-transport}
		\left\{
		\begin{aligned}
			\partial_t f^{\hbar}_\pm(t,\bx,\bk) \pm c\,\hat\bk \cdot \nablax f^{\hbar}_\pm(t,\bx,\bk)
			- \Big[ \nablax  V_\Gamma +\nablax K * \rho^\hbar \Big] \cdot \nablak f^{\hbar}_\pm(t,\bx,\bk) =\;& \Err_\pm^{\hbar}(t,\bx,\bk)\\[5pt]
			f^{\hbar}_\pm(t=0,\bx,\bk) =& \Tr[\Pi_\pm \cW^{\hbar}_0](\bx,\bk),
		\end{aligned}
		\right.
	\end{equation}
	where $\hat\bk$ is the massless group velocity defined as $\hat\bk=\bk/|\bk|$ and
	\begin{equation*}
		\rho^\hbar(t,\bx)
		:= \int_{\R^2} \big( f^\hbar_+(t,\bx,\bk) + f^\hbar_-(t,\bx,\bk) \big)\,\dd\bk,
	\end{equation*}
	furthermore, the total error term $\Err_\pm^{\hbar}$ satisfies, for $\eta (\bx) \in W^{1,\infty}_\bx(\R^2)$ and $ \phi (\bk) \in W^{2,\infty}_\bk(\Omega_\kappa)$ for any $\kappa > 0$,
	\begin{equation}\label{est-totoal-massless}
		\sup_{t\in[0,T]}\ \sup_{\substack{ \|\eta\|_{W^{1,\infty}_\bx} \le 1\\ 
				\|\phi\|_{W^{2,\infty}_\bk(\Omega_\kappa)} \le 1}}
		\big|\langle \Err_\pm^{\hbar}(t,\bx,\bk), \eta(\bx) \phi(\bk)\rangle \big|
		\ \le\ C_T\, \hbar \left( 
		1 + \ve^{-2}\right).
	\end{equation}
	
	Consequently, as $\hbar \to 0$, $f^{\hbar}_\pm$ converges to $f_\pm$ in
	$C \big([0,T]; (W^{1,\infty}_\bx W^{2,\infty}_\bk(\Omega_\kappa))'\big)$.
	The limit functions $f_\pm$ are the solutions to the Vlasov equation
	\begin{equation}\label{eq:massless-limit}
		\left\{
		\begin{aligned}
			\partial_t f_\pm(t,\bx,\bk)
			\;\pm\; c\,\hat\bk \cdot \nablax f_\pm(t,\bx,\bk)
			\;-\; \nablax V(t,\bx) \cdot \nablak f_\pm(t,\bx,\bk) &\;=\; 0,\\[4pt]
			f_\pm(t=0,\bx,\bk)
			&\;=\; \Tr\big[\Pi_\pm(\bk)\,\cW_0(\bx,\bk)\big],
		\end{aligned}
		\right.
	\end{equation}
	where
	\[
	\Tr\big[\Pi_\pm(\bk)\,\cW_0(\bx,\bk)\big]
	= \lim_{\hbar \to 0} \Tr\big[\Pi_\pm(\bk)\,\cW_0^\hbar(\bx,\bk)\big],
	\]
	and the mean-field potential is given by
	\begin{equation*}
		V(t,\bx)
		= V_\Gamma\!\left(\frac{\bx}{\ve}\right)
		+ (K * \rho)(t,\bx),
		\qquad
		\rho(t,\bx)
		:= \lim_{\hbar\to0}\int_{\R^2} \big( f^\hbar_+(t,\bx,\bk) + f^\hbar_-(t,\bx,\bk) \big)\,\dd\bk.
	\end{equation*}
\end{theorem}

\medskip

\begin{remarks*}\phantom{ }
	\begin{enumerate}
		\item The derived effective equations are physically intuitive: electrons in conduction and valence propagate in opposite directions according to their respective group velocities $v_\pm(\bk)=\pm v(\bk)$.
		
		\item In the massless case, the band projectors $\Pi_\pm(\bk)=\frac{1}{2}(\Id\pm\boldsymbol{\alpha}\!\cdot\!\hat{\bk})$ are singular at the Dirac point $\bk=\bk_*$. To control this singularity, one restricts the analysis to $|\bk-\bk_*|>\kappa$. This cutoff can be interpreted as excluding long-range modulations in the initial data: the corresponding characteristic modulation scale is of order $1/\kappa$, which may become very small depending on~$\hbar$.
		
		\item In contrast, for the massive case ($m>0$), the spectral gap $2mc^2$ regularizes the projectors and removes the singularity at $\bk=\bk_*$. Consequently, all estimates remain uniform in $\bk$, and no cutoff is required to obtain global error bounds.
	\end{enumerate}
\end{remarks*}

\medskip

As a byproduct, we can apply our results to the dynamics of the 2D material graphene with the regularized Coulomb kernel, where the periodic potential is given in the honeycomb lattice .
We recall that the 3D Coulomb kernel restricted to $\mathbb{R}^2$ is $K(\bx)=|\bx|^{-1}$, which is too singular.
To regularize it, we convolve the kernel with a Gaussian function
\begin{equation}
	G_\sigma(\bx)=\frac{1}{2\pi\sigma^{2}}\,\ee^{-\frac{|\bx|^{2}}{2\sigma^{2}}}.
\end{equation}
A direct Fourier calculation shows
\[
\widehat{K * G_\sigma}(\bk)
= \widehat{K}(\bk)\,\widehat{G_\sigma}(\bk)
= \frac{1}{|\bk|}\,\ee^{-\tfrac12\sigma^{2}|\bk|^{2}}.
\]
Because we cannot bound $\|\psi^\hbar\|_{H^1_\bx}$ uniformly in $\sigma$, we will instead use $\|\psi^\hbar\|_{L^2_\bx}$. To do so, we have the following estimate that
\[
\left\Vert  \bk'^{2} \widehat{K*G_\sigma}(\bk') \hat{\rho}(\bk') \right\Vert _{L^1_{\bk'}}
\leq \left\Vert  \bk'^{2} \widehat{K*G_\sigma}(\bk') \right\Vert _{L^1_{\bk'}} \left\| \hat{\rho}(\bk') \right\Vert _{L^\infty_{\bk'}},
\]
furthermore, we find
\[
\bigl\|\,|\bk|^2\,\widehat{K*G_\sigma}(\bk)\,\bigr\|_{L^{1}_{\bk}}
= \int_{\mathbb{R}^{2}}\,|\bk|^2\,\ee^{-\tfrac12\sigma^{2}|\bk|^{2}}\,\dd\bk
= \frac{2}{\sigma^{4}},
\]
which is finite for all $\sigma>0$.  
Thus, following the similar proof as in the main theorems, the corresponding error term the total error term $\Err_\pm^{\hbar}$ in this case converges to $0$, whenever $\sigma=h^\alpha$ with $0\leq \alpha <\frac{1}{4}$.

We summarize the application to the regularized Coulomb potential in the following Corollary.

\medskip

\begin{corollary}
	Under the conditions of Theorem \ref{thm:massiv}, let $\psi^{\hbar}$ be the solution to the scaled Dirac equation \eqref{eq:massive-dirac-scaled} with
	$K_\alpha=G_{\hbar^{\alpha}}*\frac{1}{|\;\cdot\;|}(x)$ with $0\leq \alpha <\frac{1}{4}$.
	Then, $f_\pm^\hbar$ satisfies \eqref{eq:massive-eff-transport} with
	$K_{\alpha}:=G_{\hbar^{\alpha}}*\frac{1}{|\;\cdot\;|}(x)$ with $0\leq \alpha <\frac{1}{4}$.
	Furthermore, as $\hbar \to 0 $, $f^{\hbar}_\pm$ converges to $f_\pm$ in $C\big([0,T]; (W^{1,\infty}_\bx W^{1,\infty}_\bk)' \big)$. The limit functions $f_\pm$ are the solutions to the Vlasov equation with regularized Coulomb interactive potential,
	\begin{equation*}
		\left\{
		\begin{aligned}
			\partial_t f_\pm(t,\bx,\bk) \pm \mathbf v(\bk) \cdot \nablax f_\pm(t,\bx,\bk) - \nablax V_\alpha \cdot \nablak f_\pm(t,\bx,\bk) &= 0,\\[4pt]
			f_\pm(t=0,\bx,\bk) &= \Tr[\Pi_\pm \cW_0](\bx,\bk),
		\end{aligned}
		\right.
	\end{equation*}
	where the regularized potential is given by
	\begin{equation*}
		V_\alpha(t,\bx)
		= V_\Gamma\!\left(\frac{\bx}{\ve}\right)
		+ (K_\alpha * \rho)(t,\bx),
		\qquad
		\rho(t,\bx)
		:= \lim_{\hbar\to0}\int_{\R^2} \big( f^\hbar_+(t,\bx,\bk) + f^\hbar_-(t,\bx,\bk) \big)\,\dd\bk.
	\end{equation*}
\end{corollary}
\medskip

\begin{remark}[Nonrelativistic limit of the massive case]\label{rem:semirel}
	For fixed mass $m>0$, it is natural to consider the non-relativistic limit $c\to\infty$.  
	Under suitable regularity and scaling assumptions, the solutions $f_\pm$ of the effective transport equation~\eqref{eq:massive-eff-transport} are expected to satisfy
	\begin{equation}\label{eq:Vlasov}
		\partial_t f_\pm
		\pm \frac{\bk}{m}\!\cdot\!\nabla_\bx f_\pm
		+ [\nabla_\bx V(\bx)]\!\cdot\!\nabla_\bk f_\pm
		= R_\pm(t,\bx,\bk),
	\end{equation}
	where the residual satisfies
	\[
	R_\pm(t,\bx,\bk) \;\longrightarrow\; 0
	\qquad \text{as } c\to\infty.
	\]
	In particular, in the non-relativistic limit one expects convergence (in an appropriate topology) to the classical Vlasov equation
	\begin{equation}\label{eq:nonrel-limit}
		\partial_t f_\pm
		\pm \frac{\bk}{m}\!\cdot\!\nabla_\bx f_\pm
		+ [\nabla_\bx V(\bx)]\!\cdot\!\nabla_\bk f_\pm
		= 0.
	\end{equation}
	Rigorous results in three dimensions can be found in~\cite{HongPankavich2025}.
	
	In our setting, however, certain kinetic error terms contain powers of $c$ that diverge as $c\to\infty$, so the non-relativistic limit is not directly accessible from our present estimates.  
	To the best of our knowledge, establishing a combined semiclassical and non-relativistic limit for the Dirac--Hartree dynamics therefore remains an interesting open problem.
	
	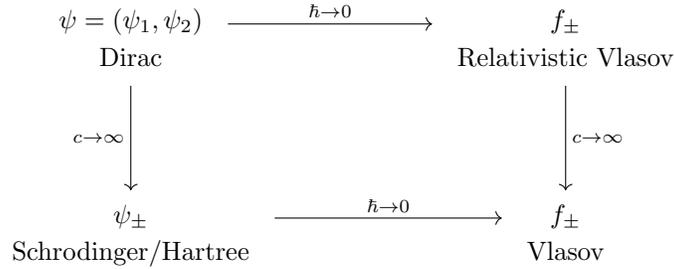
\begin{figure}[ht]
		\centering
		\[
		\begin{tikzcd}[column sep=6em, row sep=4em]
			{\raisebox{-1.5ex}{$\begin{array}{c} \psi=(\psi_1,\psi_2) \\[0.5ex] \text{Dirac} \end{array}$}} 
			\arrow[r, "\hbar \to 0"] \arrow[d, "c \to \infty"'] 
			& {\raisebox{-1.5ex}{$\begin{array}{c} f_\pm \\[0.5ex] \text{Relativistic Vlasov} \end{array}$}} 
			\arrow[d, "c \to \infty"] \\
			{\raisebox{-1.5ex}{$\begin{array}{c} \psi_\pm \\[0.5ex] \text{Schrodinger/Hartree} \end{array}$}} 
			\arrow[r, "\hbar \to 0"] 
			& {\raisebox{-1.5ex}{$\begin{array}{c} f_\pm \\[0.5ex] \text{Vlasov} \end{array}$}}
		\end{tikzcd}
		\]
		\caption{Diagram of the limiting process $\hbar \to 0$ and $c \to \infty$ for the massive case.}
		\label{fig:comm-diagram}
	\end{figure}
\end{remark}

\begin{remark}[Homogenization limit]
	In the formal regime $\ve \to 0$, if the periodic potential has zero cell average, i.e., $\int_{\mathcal C} V_\Gamma = 0$, then the fast oscillations of $V_\Gamma(x/\ve)$ are expected to average out, and the macroscopic band distributions $f_\pm$ should satisfy the Vlasov equation without an external potential.  
	This corresponds to the \emph{homogenization limit} for periodic media; see, for instance, \cite{GMMP1997}.
	However, in our setting, the kinetic error term \eqref{eq:weak-kinetic-merged} involves negative powers of $\ve$ within a time-dependent exponential factor.  
	As a consequence, the corresponding bounds are not uniform in the limit $\ve\to0$, and a homogenization limit cannot be extracted within our framework.
	By contrast, the two potential contributions remain uniformly controlled as $\ve\to0$ and do not exhibit such growth.
	A rigorous treatment of the homogenization limit in this context is left for future work.
\end{remark}

\section{Proof of Theorems \ref{thm:massiv} and \ref{thm:massless}} 
\label{sec:massive}

Applying the projection $\Tr [\Pi_\pm (\cdot)]$ to the Wigner equation \eqref{eq:Wigner}, we have
\begin{multline}
	\partial_t f_\pm^{\hbar}
	+ c \Tr\big[\Pi_\pm (\boldsymbol{\alpha}\!\cdot\!\nabla_{\bx} \cW^{\hbar})\big]
	+ \frac{1}{\ii\hbar}\Tr\big[\Pi_\pm[H(\bk),\cW^{\hbar}]\big]\\ - \Tr\big[\Pi_\pm ( \mathcal Q^\ve[V_\Gamma,\cW^{\hbar}])\big] - \Tr\big[\Pi_\pm ( \mathcal Q^\hbar[V_{\mathrm{int}},\cW^{\hbar}])\big] = 0,
\end{multline}
which directly reduces to
\begin{equation}\label{eq:proj-eqn}
	\partial_t f_\pm^{\hbar}
	+ c \Tr\big[\Pi_\pm (\boldsymbol{\alpha}\!\cdot\!\nabla_{\bx} \cW^{\hbar})\big] - \Tr\big[\Pi_\pm ( \mathcal Q^\ve[V_\Gamma,\cW^{\hbar}])\big] - \Tr\big[\Pi_\pm ( \mathcal Q^\hbar[V_{\mathrm{int}},\cW^{\hbar}])\big] = 0,
\end{equation}
by noticing the fact that the commutator term vanishes, i.e.,
\[
\Tr\big[\Pi_\pm[H,\cW^{\hbar}]\big]=\Tr\big[ [\Pi_\pm,H]\,\cW^{\hbar}\big]=0,
\]
as $\Pi_\pm$ commutes with $H(\bk)$.

Let us define the kinetic part in \eqref{eq:proj-eqn} as
\begin{equation}\label{def:kineticpart}
	\mathcal{K}_\pm^{\hbar} := c\,\Tr\!\big[\Pi_\pm(\balpha\!\cdot\!\nabla_\bx \cW^{\hbar})\big],
\end{equation}
and, for a class of ``suitable" test functions $\eta(\bx)$ and $\phi(\bk)$, we rewrite 
\begin{equation}\label{eq:error-kin}
	\langle \mathcal{K}_\pm^{\hbar}, \eta(\bx)\phi(\bk) \rangle
	= \,\langle \bv(\bk)\!\cdot\!\nabla_\bx f^{\hbar}_\pm, \eta(\bx)\phi(\bk)\rangle \;+\; \Err^{\mathrm{kin}}_\pm[\eta,\phi],
\end{equation}
where $\bv(\bk)=\tfrac{c^2}{E_m(\bk)}\bk$ in the massive case and $\bv(\bk)=c\,\hat \bk$ in the massless case.

We can also similarly define the corresponding error terms that are related with external potential in \eqref{eq:proj-eqn}, 
\begin{equation}\label{def:error-Qext}
	\langle \Tr\big[\Pi_\pm ( \mathcal Q^\ve[V_\Gamma,\cW^{\hbar}])\big], \eta(\bx)\,\phi(\bk)  \rangle = \langle \nablax V_\Gamma\cdot\nablak f^{\hbar}_\pm, \eta(\bx)\,\phi(\bk) \rangle  +  \Err^{\mathrm{ext}}_{\pm}[\eta, \phi],
\end{equation}
and with the Hartree-type interactive potential in \eqref{eq:proj-eqn}, 
\begin{equation}\label{def:error-Qint}
	\langle \Tr\big[\Pi_\pm ( \mathcal Q^\hbar[V_{\mathrm{int}},\cW^{\hbar}])\big], \eta(\bx)\,\phi(\bk)  \rangle = \langle \nablax V_{\mathrm{int}} \cdot\nablak f^{\hbar}_\pm, \eta(\bx)\,\phi(\bk) \rangle  +  \Err^{\mathrm{int}}_{\pm}[\eta, \phi].
\end{equation}

Hence, the total error term $\Err_\pm^{\hbar}(t,\bx,\bk)$ includes
\begin{equation}
	\langle \Err_\pm^{\hbar}(t,\bx,\bk), \eta(\bx) \phi(\bk) \rangle := \Err^{\mathrm{kin}}_\pm[\eta,\phi] + \Err^{\mathrm{ext}}_{\pm}[\eta, \phi] + \Err^{\mathrm{int}}_{\pm}[\eta, \phi].
\end{equation}

In the following subsections \footnote{Throughout this section, we denote $A \lesssim B$ if there exists $C>0$ which is independent of $\hbar$ such that $A \leq C B$, and $A \sim B$ if $A \lesssim B$ and $B \lesssim A$.}, we provide detailed estimates for the error terms appearing in \eqref{eq:error-kin}, \eqref{def:error-Qext}, and \eqref{def:error-Qint}, respectively.

\subsection{Estimate of the error term from the kinetic part}

In this subsection, we quantify the contribution and error terms from the kinetic part. Meanwhile, we also illustrate that the dominant contributions arise from the diagonal component $\cW^{\hbar}_{\mathrm{D}}$; in contrast, the off-diagonal component $\cW^{\hbar}_{\mathrm{OD}}$ contributes only to the higher-order remainder terms, which are shown to be generically small.

\begin{proposition}[Estimate of the kinetic error]
	\label{prop:weak-kinetic}
	Under assumptions \textnormal{\textbf{(A1)}}--\textnormal{\textbf{(A3)}} (and \textnormal{\textbf{(A4)}} for the massless case), for all $\eta\in W_\bx^{1,\infty}(\R^2)$ and $\phi\in H^1_\bk(\R^2)$ (with $\operatorname{supp}(\phi) \subset \Omega_\kappa$ if $m=0$), 
	\begin{equation}\label{pairing-kin}
		\langle \mathcal{K}_\pm^\hbar,\eta\,\phi \rangle
		=
		\langle \bv(\bk)\cdot\nabla_\bx f_\pm^\hbar,\eta\,\phi \rangle
		+ \Err_\pm^{\mathrm{kin}}[\eta,\phi],
	\end{equation}
	where
	\[
	\langle A,B\rangle
	:= \int_{\R^2}\int_{\R^2} \Tr\bigl[ A^\dagger(\bx,\bk) B(\bx,\bk) \bigr]\,\dd\bx\,\dd\bk,
	\]
	and the group velocity is $\bv(\bk)=\nabla_\bk E_m(\bk)$. The error term $\Err_\pm^{\mathrm{kin}}$ satisfies
	\begin{equation}\label{eq:weak-kinetic-merged}
		|\Err_\pm^{\mathrm{kin}}[\eta,\phi]|
		\le C_T\,\hbar\, \|\eta\|_{W^{1,\infty}_\bx }\|\phi\|_{H^1_\bk},
	\end{equation}
	where $C_T$ depends on $T$, the spectral gap, and on the potential norms, but is independent of~$\hbar$.
\end{proposition}

\noindent\emph{Proof.}
\textbf{Step 1: Decomposition.}
We decompose $\cW^{\hbar} = \cW^{\hbar}_{\mathrm{D}} + \cW^{\hbar}_{\mathrm{OD}}$ into diagonal and off-diagonal parts. For the diagonal part $\cW^{\hbar}_{\mathrm{D}}$, we first consider the key identity
\begin{equation}\label{eq:Pi-alpha-identity}
	\Pi_\pm\alpha_j\Pi_\pm = \frac{1}{c} v_j(\bk)\,\Pi_\pm,
	\qquad j=1,2,
\end{equation}
which follows from the explicit form of the spectral projectors \eqref{eq:proj-general} and the relation between the Dirac matrices $\balpha$ and the group velocity. Here, in the massive case,
\[
v_j(\bk) = \frac{c^2\,k_j}{E_m(\bk)},
\qquad
\text{so that}\quad
\bv(\bk) = \frac{c^2\,\bk}{E_m(\bk)},
\]
and in the massless case,
\[
v_j(\bk) = c\,\frac{k_j}{|\bk|},
\qquad
\text{so that}\quad
\bv(\bk) = c\,\hat{\bk}.
\]

Using \eqref{eq:Pi-alpha-identity}, the diagonal component yields
\begin{align*}
	\Tr[\Pi_\pm\balpha\,\cW^{\hbar}_{\mathrm{D}}]
	&= \sum_{j=1}^2 \Tr[\Pi_\pm\alpha_j\Pi_\pm\,(\Pi_\pm\cW^{\hbar}\Pi_\pm)]
	+ \sum_{j=1}^2 \Tr[\Pi_\pm\alpha_j\Pi_\mp\,(\Pi_\mp\cW^{\hbar}\Pi_\mp)]\\
	&= \frac{1}{c}\sum_{j=1}^2 v_j(\bk)\,\Tr[\Pi_\pm\cW^{\hbar}]
	= \frac{1}{c}\bv(\bk)\cdot\mathbf{e}_j\,f_\pm^{\hbar},
\end{align*}
where $f_\pm^{\hbar}:=\Tr[\Pi_\pm\cW^{\hbar}]$ is the projected density \eqref{def:fpm}. Here, we used that $\Pi_\pm\alpha_j\Pi_\mp = 0$ when restricted to the diagonal blocks.

Therefore, by doing integration by part again, we can obtain
\begin{equation}
	\langle \mathcal{K}_\pm^{\hbar},\eta(\bx)\,\phi(\bk) \rangle
	=\big\langle \bv(\bk)\cdot\nabla_\bx f^{\hbar}_\pm,\eta(\bx)\phi(\bk)\big\rangle
	+\Err^{\mathrm{kin}}_\pm[\eta,\phi],
\end{equation}
where the error term $\Err^{\mathrm{kin}}_\pm$ originates from the off-diagonal part $\cW^{\hbar}_{\mathrm{OD}}$ of $\cW^{\hbar}$:
\begin{equation}\label{eq:error-kin-def}
	\Err^{\mathrm{kin}}_\pm[\eta,\phi]
	= -\,c\,\big\langle
	\Tr[\Pi_\pm\balpha\,\cW^{\hbar}_{\mathrm{OD}}],
	\nabla_\bx\eta(\bx)\,\phi(\bk)
	\big\rangle.
\end{equation}

\textbf{Step 2: Interband smallness.}
To bound the error term \eqref{eq:error-kin-def}, we rely on adiabatic decoupling between the positive and negative energy bands. A direct energy estimate for the off-diagonal evolution does not, by itself, yield $\mathcal{O}(\hbar)$ bound, since the source term contains $\mathcal{O}(1)$ contribution coming from the Berry connection. However, the rapid oscillation of the phase $\ee^{\pm 2 \ii E_m t/\hbar}$ implies that the net effect is small on time scales of order one. The following lemma is used to illustrate the "small" contribution of the off-diagonal part.

\begin{lemma}[Adiabatic decoupling]\label{lem:OD-L1H-1}
	Under assumptions \textnormal{\textbf{(A1)}}--\textnormal{\textbf{(A3)}} (and \textnormal{\textbf{(A4)}} if $m=0$), there exists a constant $C_T$ independent of $\hbar$ such that
	\begin{equation}\label{eq:L2-OD-est}
		\sup_{t\in[0,T]} \|\cW_{\mathrm{OD}}^\hbar(t)\|_{L^2(\R^2_\bx \times \R^2_\bk )} \le C_T \, \hbar.
	\end{equation}
	Consequently, for any $\eta \in L^2_\bx$ and $\phi \in L^2_\bk$,
	\begin{equation}\label{eq:weak-OD-pairing}
		|\langle \cW_{\mathrm{OD}}^\hbar(t), \eta\,\phi \rangle|
		\le C_T \, \hbar \, \|\eta\|_{L^2_\bx}\|\phi\|_{L^2_\bk}.
	\end{equation}
\end{lemma}

\begin{proof}[Sketch of Proof]
    We briefly recall the argument and refer to \cite{PanatiSpohnTeufel2003,Teufel2003} for details.
    
    We remark that although the adiabatic theory in \cite{Teufel2003} is formulated for linear Hamiltonians, it applies here via a self-consistent argument. We treat the Hartree potential $V_{\mathrm{int}}(t,\bx)$ as an effective time-dependent external potential. Thanks to the regularity of the kernel $K$ in \textnormal{\textbf{(A1)}} and the conservation laws, the time derivative $\partial_t V_{\mathrm{int}} = (\nabla K) * \mathbf{J}$ remains uniformly bounded (where $\mathbf{J} = \langle \psi^\hbar, \balpha \psi^\hbar \rangle$ is the current density). Thus, the effective time-dependent Hamiltonian satisfies the slow-variation regularity conditions required for the linear space-adiabatic expansion.

    The Hamiltonian $H_m(\bk)$ has a spectral gap separating the positive and negative energy bands (bounded away from zero by $2m c^2$ in the massive case, or by $2c\kappa$ in the massless case under \textnormal{\textbf{(A4)}}). Space-adiabatic perturbation theory shows that, for each $N\in\mathbb N$, one can construct a ``super-adiabatic'' family of projectors $P_N^\hbar$, possibly depending on time $t$ but does not diverge for $t\in[0,T]$, such that
    \[
    \bigl\|[H_m + V(t),P_N^\hbar]\bigr\|_{\mathcal{B}(L^2)}
    \le C_{N,T}\,\hbar^N
    \qquad\text{for } t\in[0,T],
    \]
    and $P_N^\hbar$ approximates the spectral projectors $\Pi_\pm$ with
    \[
    \bigl\|P_N^\hbar - \Pi_\pm\bigr\|_{\mathcal{B}(L^2)}
    \le C_{N,T}\,\hbar.
    \]
    For initial data satisfying the well-preparedness condition \textnormal{\textbf{(A3)}}, the off-diagonal component with respect to $P_N^\hbar$ remains of order $\hbar^N$ for $t\in[0,T]$. Since $P_N^\hbar$ and $\Pi_\pm$ differ by $\mathcal{O}(\hbar)$ in operator norm, the corresponding off-diagonal component with respect to the fixed projectors $\Pi_\pm$ is of order $\hbar$. Passing from the density matrix to its Wigner transform preserves the Hilbert--Schmidt norm (up to a constant factor), hence, the estimate \eqref{eq:L2-OD-est} follows, and \eqref{eq:weak-OD-pairing} is an immediate consequence of the Cauchy--Schwarz inequality.
\end{proof}

\textbf{Step 3: Conclusion.}
We estimate the $L^2$–norm of the test function $\eta(\bx)\phi(\bk)$. Since $\|\Pi_\pm(\bk)\|$ and $\|\balpha\|$ are uniformly bounded and $\nabla_\bx \eta\in L^\infty_\bx$,
applying Lemma~\ref{lem:OD-L1H-1} with $\Psi=\tilde{\Phi}$ gives
\[
|\Err_\pm^{\mathrm{kin}}[\eta,\phi]|
= |\langle \cW_{\mathrm{OD}}^\hbar,\tilde{\Phi}\rangle|
\le C_T\,\hbar\,\|\tilde{\Phi}\|_{L^2}
\le C_T\,\hbar\,
\|\eta\|_{W^{1,\infty}_\bx}\,\|\phi\|_{H^1_\bk},
\]
which is \eqref{eq:weak-kinetic-merged}.
\qed

\subsection{Estimate of error term from external potential part}

In this subsection, we illustrate that the external potential related term $\Tr\big[\Pi_\pm ( \mathcal Q^\ve[V_\Gamma,\cW^{\hbar}])\big]$ weakly converges to the force term $-\nablax V_\Gamma\cdot\nablak f^{\hbar}_\pm$, similar to the semiclassical limit in the case of the Schr\"odinger dynamics. This behavior remains unchanged in the Dirac case, and the following estimate is presented for clarity.

\begin{proposition}[Estimate of external potential related error]\label{prop:error-Vext}
	For the external potential $V_\Gamma$ in \textnormal{\textbf{(A1)}} with $\ve\in(0,1]$, let $\psi^{\hbar}$ be the solution to \eqref{eq:massive-dirac-scaled} (or \eqref{eq:massless-dirac-scaled}) satisfying
	the assumptions in Theorem \ref{thm:massiv} (or \ref{thm:massless}). 
	Then, there exists a constant $C$ independent of $\hbar$ such that, for any $\eta\in L^\infty_\bx$ and 
	$\phi\in W^{2,\infty}_\bk$,
	\begin{equation}\label{eq:potential-weak-massive}
		\langle \Tr\big[\Pi_\pm ( \mathcal Q^\ve[V_\Gamma,\cW^{\hbar}])\big], \, \eta(\bx)\,\phi(\bk)  \rangle = \langle \nablax V_\Gamma\cdot\nablak f^{\hbar}_\pm, \,\eta(\bx)\,\phi(\bk) \rangle  +  \Err^{\mathrm{ext}}_{\pm}[\eta, \phi],
	\end{equation}
	with
	\begin{equation}\label{est:error-Vext}
		\left| \Err^{\mathrm{ext}}_{\pm}[\eta, \phi] \right| \;\le\; C \,\frac{\hbar}{\ve^2}\,\|\eta\|_{L^\infty_\bx}\,\|\phi\|_{W^{2,\infty}_\bk},
	\end{equation}
	uniformly in $t\in[0,T]$.
	Consequently, we have
	\begin{equation}
		\Tr\big[\Pi_\pm ( \mathcal Q^\ve[V_\Gamma,\cW^{\hbar}])\big] \to \nablax V_\Gamma\cdot\nablak f^{\hbar}_\pm \quad \text{as} \quad  \hbar \to 0,
		\quad \text{in } C\big([0,T]; (L^\infty_\bx W^{2,\infty}_\bk)' \big).
	\end{equation}
\end{proposition}

\begin{proof}
	
	We recognize the contribution of the external-related term by using the following formula of $\mathcal{Q}[V_\Gamma, \cW^{\hbar}]$ as in the last line of \eqref{eq:Q-L}: for $\hbar/\ve^2 \ll 1$,
	\begin{equation}\label{Q-ext-gamma}
		\begin{aligned}
			\mathcal Q^\ve[V_\Gamma, \cW^{\hbar}](t,\bx,\bk)= &\frac{\ii}{\hbar} \sum_{\bmu \in \Gamma^*} \ee^{\ii \bmu \cdot \frac{\bx}{\ve}} \hat{V}_{\Gamma}(\bmu) \left[ \cW^{\hbar} \left(t,\bx, \bk +\frac{\hbar}{2\ve}\bmu \right) - \cW^{\hbar} \left(t,\bx, \bk -\frac{\hbar}{2\ve}\bmu \right)  \right]\\[5pt]
			=& \frac{\ii}{\hbar} \sum_{\bmu \in \Gamma^*} \ee^{\ii \bmu \cdot \frac{\bx}{\ve}} \hat{V}_{\Gamma}(\bmu) \left[  \frac{\hbar}{\ve} \bmu\cdot \nablak \cW^{\hbar} (t,\bx,\bk) + \mathcal{O}\left( \frac{\hbar^2}{\ve^2} \right) \right]\\[5pt]
			=&  \nabla_\bx V_\Gamma \cdot\nabla_\bk \cW^{\hbar}(t,\bx,\bk) + \mathcal{O}\left( \frac{\hbar}{\ve^2} \right),
		\end{aligned}
	\end{equation}
	where we substitute $\ii\sum_{\bmu \in \Gamma^*} \ee^{\ii \bmu \cdot \frac{\bx}{\ve}} \hat{V}_{\Gamma}(\bmu) \bmu = \ve\nablax V_\Gamma$ in the last equality above.
	
	In what follows, we will rigorously justify the error terms under $\Tr \Pi_\pm$ is $\mathcal{O}\left( \frac{\hbar}{\ve^2}\right)$ in the weak sense:
	we apply the Taylor expansion to $\Pi_{\pm}\cW^\hbar$ that
	\begin{multline*}
		\Pi_{\pm}\cW^\hbar\left(t,\bx,\bk+\frac{\hbar}{2\ve}\bmu \right) - \Pi_{\pm}\cW^\hbar \left(t,\bx,\bk-\frac{\hbar}{2\ve}\bmu\right)\\[4pt]
		= \frac{\hbar}{\ve}\bmu\,\nabla_{\bk}\Pi_{\pm}\cW^\hbar(t,\bx,\bk)
		+\frac{\hbar^{2}}{\ve^{2}}\int_{0}^{1}\!\!\int_{0}^{1}2(2s-1)\bmu^{2}\nabla_{\bk}^{2}\Pi_{\pm}\cW^\hbar \left(t,\bx,\bk+\frac{\theta(2s-1)}{2\ve}\bmu\right) \,\mathrm{d}\theta \mathrm{d}s.
	\end{multline*}
	Then, by taking the trace operator $\Tr$, we have
	\begin{multline}\label{eq:tr-W-diff}
		\Tr\left[\Pi_{\pm}\cW^\hbar\left(t,\bx,\bk+\frac{\hbar}{2\ve}\bmu\right) - \Pi_{\pm}\cW^\hbar\left(t,\bx,\bk-\frac{\hbar}{2\ve}\bmu\right) \right]
		= \frac{\hbar}{\ve}\bmu\,\nabla_{\bk}\Tr\left[\Pi_{\pm}\cW^\hbar(t,\bx,\bk)\right]\\[4pt]
		+\frac{\hbar^{2}}{\ve^{2}}\int_{0}^{1}\!\!\int_{0}^{1}2(2s-1)\bmu^{2}\Tr\left[\nabla_{\bk}^{2}\Pi_{\pm}\cW^\hbar \left(t,\bx,\bk+\frac{\theta(2s-1)}{2\ve}\bmu\right)\right]\,\mathrm{d}\theta \mathrm{d}s.
	\end{multline}
	
	Hence, by taking $\Tr \Pi_\pm$ to \eqref{Q-ext-gamma} and substituting \eqref{eq:tr-W-diff}, we can obtain the expansion as in \eqref{def:error-Qext} that
	\begin{equation*}
		\langle \Tr\big[\Pi_\pm ( \mathcal Q^\ve[V_\Gamma,\cW^{\hbar}])\big], \eta(\bx)\,\phi(\bk)  \rangle = \langle \nablax V_\Gamma\cdot\nablak f^{\hbar}_\pm, \eta(\bx)\,\phi(\bk) \rangle  +  \Err^{\mathrm{ext}}_{\pm}[\eta, \phi],
	\end{equation*}
	where
	\begin{multline*}
		\Err^{\mathrm{ext}}_{\pm}[\eta, \phi]
		=\\
		\frac{\ii}{\hbar}\sum_{\bmu\in\Gamma^{*}}\ee^{\ii\bmu\cdot\frac{\bx}{\ve}}\hat{V}_{\Gamma}(\bmu)\left[\frac{\hbar^{2}}{\ve^{2}}\int_{0}^{1}\!\!\int_{0}^{1}2(2s-1)\bmu^{2}\left\langle \Tr\left[\nabla_{\bk}^{2}\Pi_{\pm}\cW^\hbar \left(t,\bx,\bk+\frac{\theta(2s-1)}{2\ve}\bmu\right) \right],\eta(\bx)\phi(\bk)\right\rangle \,\mathrm{d}\theta\mathrm{d}s\right]
	\end{multline*}
	We further estimate the external potential related error term,
    \begin{equation}\label{est:Q-ext-error}
	\begin{aligned}
		&\left|\Err^{\mathrm{ext}}_{\pm}[\eta, \phi] \right|\\
		&\lesssim \frac{\hbar}{\ve^{2}}\Bigg|\sum_{\bmu\in\Gamma^{*}}\ee^{\ii\bmu\cdot\frac{\bx}{\ve}}\hat{V}_{\Gamma}(\bmu)\bmu^{2}\int_{0}^{1}\!\!\int_{0}^{1}2(2s-1)\left\langle \Tr\left[\nabla_{\bk}^{2}\Pi_{\pm}\cW^\hbar\left(t,\bx,\bk+\frac{\theta(2s-1)}{2\ve}\bmu\right)\right],\eta(\bx)\phi(\bk)\right\rangle \,\,\mathrm{d}\theta\mathrm{d}s\Bigg|\\[4pt]
		&= \frac{\hbar}{\ve^{2}}\Bigg|\sum_{\bmu\in\Gamma^{*}}e^{\ii\bmu\cdot\frac{\bx}{\ve}}\hat{V}_{\Gamma}(\bmu)\bmu^{2}\int_{0}^{1}\!\!\int_{0}^{1}2(2s-1)\left\langle \Tr\left[\Pi_{\pm}\cW^\hbar \left(t,\bx,\bk+\frac{\theta(2s-1)}{2\ve}\bmu\right) \right],\eta(\bx)\nabla_{\bk}^{2}\phi(\bk)\right\rangle \,\mathrm{d}\theta\mathrm{d}s\Bigg|\\[4pt]
		&\lesssim \frac{\hbar}{\ve^{2}}\left\Vert \bmu^{2}\hat{V}_{\Gamma}(\bmu)\right\Vert _{\ell_{\bmu}^{1}} 
		\int_{0}^{1}\!\!\int_{0}^{1} \left\Vert \Tr\left[\Pi_{\pm}\cW^\hbar \left(t,\bx,\bk+\frac{\theta(2s-1)}{2\ve}\bmu\right) \right]\right\Vert_{L_{\bx}^{1}L_{\bk}^{1}}\!\dd \theta \dd s\;
		\left\Vert \eta(\bx)\right\Vert _{L_{\bx}^{\infty}}\left\Vert \nabla_{\bk}^{2}\phi(\bk)\right\Vert _{L_{\bk}^{\infty}}\\[4pt]
		&= \frac{\hbar}{\ve^{2}}
		\Big\Vert \bmu^{2} \hat{V}_{\Gamma}(\bmu) \Big\Vert_{\ell_{\bmu}^{1}}
		\left\Vert \Tr\left[\Pi_{\pm}\cW^\hbar(t,\bx,\bk)\right] \right\Vert_{L_{\bx}^{1}L_{\bk}^{1}}
		\left\Vert \eta(\bx) \right\Vert_{L_{\bx}^{\infty}}
		\left\Vert \nabla_{\bk}^{2}\phi(\bk) \right\Vert_{L_{\bk}^{\infty}}\\[4pt]
		&\leq C\frac{\hbar}{\ve^{2}}
		\left\Vert \eta(\bx) \right\Vert_{L_{\bx}^{\infty}}
		\left\Vert \phi(\bk) \right\Vert_{W^{2,\infty}_\bk},
	\end{aligned}  
    \end{equation}
	where the H\"older inequality in $\bx,\bk$, since the integrand of $s,\theta$ is independent of all variables $\bx,\bk,\bmu$. This concludes the proof by considering $V_\Gamma \in A^2(\mathcal{C})$. 
\end{proof}

\begin{remark}
	Throughout this work, we regard the lattice scale $\ve$ as a fixed $\mathcal{O}(1)$ parameter and perform the semiclassical limit $\hbar\to0$ only.  In particular, we do \emph{not} consider the regime $\ve\to0$. The presence of the off--diagonal component $\cW_{\mathrm{OD}}^\hbar$ already suggests that sending $\ve\to0$ would require additional spectral analysis beyond the scope of this paper.
	When $\hbar/\ve^{2}\ll 1$, the external periodic potential generates, to leading order, the familiar local force term
	$\nabla_\bx V_\Gamma \cdot\nabla_\bk \cW^\hbar$,
	which does not depend on the band geometry. Consequently, the weak semiclassical limit of this force term follows the same structure as in the Schr\"odinger case; see, e.g., \cite{QiWangWatson2025}.   
	This argument applies uniformly to both the massive and the massless Dirac settings.
	
	For other ratios of $\hbar/\ve$, however, the external potential can induce genuinely nonlocal mixing effects in the momentum variable, potentially altering the effective dynamics in a nontrivial way.   A precise description of these regimes would require a finer microlocal and spectral analysis of the Dirac operator with periodic coefficients.   We do not pursue this direction here and leave it for future work.
\end{remark}

\subsection{Estimate of error term from interactive potential part}

In this subsection, we estimate the interactive potential related term $\Tr\big[\Pi_\pm ( \mathcal Q^\hbar[V_{\mathrm{int}},\cW^{\hbar}])\big]$, which is proved to converge to $-\nablax V_{\mathrm{int}} \cdot\nablak f^{\hbar}_\pm$ in a certain weak sense.

\begin{proposition}[Estimate of the Hartree-type interactive potential related error]
	\label{prop:error-Vint}
	For the interactive kernel $K$ in \textnormal{\textbf{(A1)}}, let $\psi^{\hbar}$ be the solution to \eqref{eq:massive-dirac-scaled} (or \eqref{eq:massless-dirac-scaled}) satisfying
	the assumptions in Theorem \ref{thm:massiv} (or \ref{thm:massless}).  Then, there exists a constant $C$ independent of $\hbar$ such that, for any $\eta \in L^\infty_\bx$ and $\phi \in W^{2,\infty}_\bk$,
	\begin{equation}
		\langle \Tr\big[\Pi_\pm ( \mathcal Q^\hbar[V_{\mathrm{int}},\cW^{\hbar}])\big], \eta(\bx)\,\phi(\bk)  \rangle = \langle \nablax V_{\mathrm{int}} \cdot\nablak f^{\hbar}_\pm, \eta(\bx)\,\phi(\bk) \rangle  +  \Err^{\mathrm{int}}_{\pm}[\eta, \phi],
	\end{equation}
	with
	\begin{equation}\label{est:error-Vint}
		\left| \Err^{\mathrm{int}}_{\pm}[\eta, \phi] \right| \;\le\; C \,\hbar \,\|\eta\|_{L^\infty_{\bx}}\|\phi\|_{W^{2,\infty}_{\bk}},
	\end{equation}
	uniformly in $t\in[0,T]$. Consequently, we have
	\begin{equation}
		\Tr\big[\Pi_\pm ( \mathcal Q^\hbar[V_{\mathrm{int}},\cW^{\hbar}])\big] \to \nablax V_{\mathrm{int}} \cdot\nablak f^{\hbar}_\pm \quad \text{as} \quad \hbar \to 0,
		\quad \text{in } C\big([0,T]; (L^\infty_\bx W^{2,\infty}_\bk)' \big). 
	\end{equation}
\end{proposition}

\begin{proof}
	Similar to the derivation for the external potential term, we can figure out the contribution of the Hartree interactive potential related term by using the following formula of $\mathcal Q^{\hbar}[V_{\mathrm{int}},\cW^{\hbar}]$ as in the last line of \eqref{eq:Q-int}:
	\begin{equation}\label{eq:Q-V-int-}
		\begin{aligned}
			\mathcal Q^{\hbar}[V_{\mathrm{int}},\cW^{\hbar}](t,\bx,\bk)
			=& \frac{\ii}{\hbar} \int_{\R^2}
			\int_{\R^2}
			\ee^{\ii \bk' \cdot (\bx-\by)} V_{\mathrm{int}} \left(t, \by \right) 
			\Big[ \cW^{\hbar}\left(t,\bx,\bk+ \frac{\hbar\bk'}{2}\right) - \cW^{\hbar}\left(t,\bx,\bk- \frac{\hbar\bk'}{2}\right) \Big]\dd\bk' \,\dd\by \\[5pt]
			=& \frac{\ii}{\hbar} \int_{\R^2}
			\ee^{\ii \bk' \cdot \bx} \widehat{V_{\mathrm{int}}} \left(t, \bk' \right) \left[  \hbar \bk' \nablak \cW^{\hbar} (t,\bx,\bk) + \mathcal{O}\left( \hbar^2 \right) \right]\dd\bk'\\[5pt] 
			=& \nablax V_{\mathrm{int}}(t,\bx) \cdot\nabla_\bk \cW^{\hbar}(t,\bx,\bk) + \mathcal{O} (\hbar),
		\end{aligned}
	\end{equation}
	where $V_{\mathrm{int}}(t,\bx) = K * \rho (t,\bx)$.
	
	Then, we need to rigorously justify that the error terms under $\Tr \Pi_\pm$ is $\mathcal{O}\left( \hbar \right)$ in a certain weak sense. Following the similar calculation in \eqref{eq:tr-W-diff}, we have 
	\begin{multline}\label{eq:tr-W-k'}
		\Tr\left[\Pi_{\pm}\cW^{\hbar}\left(t,\bx,\bk+ \frac{\hbar\bk'}{2}\right) - \Pi_{\pm}\cW^{\hbar}\left(t,\bx,\bk- \frac{\hbar\bk'}{2}\right) \right]
		= \hbar \bk'\,\nabla_{\bk}\Tr\left[\Pi_{\pm}\cW^\hbar(t,\bx,\bk)\right]\\[4pt]
		+\hbar^2 \int_{0}^{1}\!\!\int_{0}^{1}2(2s-1)\bk'^2\Tr\left[\nabla_{\bk}^{2}\Pi_{\pm}\cW^\hbar \left(t,\bx,\bk+\theta(2s-1)\bk'\right)\right]\,\mathrm{d}\theta \,\mathrm{d}s.
	\end{multline}
	Therefore, taking $\Tr \Pi_\pm$ to \eqref{eq:Q-V-int-} and substituting \eqref{eq:tr-W-k'}, it leads to the expansion as in \eqref{def:error-Qint},
	\begin{equation}
		\langle \Tr\big[\Pi_\pm ( \mathcal Q^\hbar[V_{\mathrm{int}},\cW^{\hbar}])\big], \eta(\bx)\,\phi(\bk)  \rangle = \langle \nablax V_{\mathrm{int}} \cdot\nablak f^{\hbar}_\pm, \eta(\bx)\,\phi(\bk) \rangle  +  \Err^{\mathrm{int}}_{\pm}[\eta, \phi],
	\end{equation}
	where the error term $\Err^{\mathrm{int}}_{\pm}[\eta, \phi]$ is given by
	\begin{align*}
		&\Err^{\mathrm{int}}_{\pm}[\eta, \phi]\\
		&=\frac{\ii}{\hbar} \int_{\R^2} \ee^{\ii\bk'\cdot\bx} \widehat{V}_{\mathrm{int}}(\bk')\left[\hbar^2 \int_{0}^{1}\!\!\int_{0}^{1}2(2s-1)\bk'^{2}\left\langle \Tr\left[\nabla_{\bk}^{2}\Pi_{\pm}\cW^\hbar\left(t,\bx,\bk+\theta(2s-1)\bk'\right)\right],\eta(\bx)\phi(\bk)\right\rangle \,\mathrm{d}\theta\mathrm{d}s\right]\,\dd\bk'.
	\end{align*}
	We further estimate the interactive potential related error term by using its Hartree type convolution, so that we have
	\begin{align*}
		&\left|\Err^{\mathrm{int}}_{\pm}[\eta, \phi] \right|\\
		&\lesssim \hbar\,\Bigg|\int_{\R^2}\ee^{\ii\bk'\cdot\bx} \widehat{V}_{\mathrm{int}}(\bk')\bk'^{2}\int_{0}^{1}\!\!\int_{0}^{1}2(2s-1)\left\langle \Tr\left[\nabla_{\bk}^{2}\Pi_{\pm}\cW^\hbar\left(t,\bx,\bk+\theta(2s-1)\bk'\right)\right],\eta(\bx)\phi(\bk)\right\rangle \,\mathrm{d}\theta\mathrm{d}s\,\dd\bk'\Bigg|\\[4pt]
		&= \hbar\,\Bigg|\int_{\R^2} \ee^{\ii\bk'\cdot\bx} \widehat{V}_{\mathrm{int}}(\bk')\bk'^{2} \int_{0}^{1}\!\!\int_{0}^{1}2(2s-1)\left\langle \Tr\left[\Pi_{\pm}\cW^\hbar\left(t,\bx,\bk+\theta(2s-1)\bk'\right)\right],\eta(\bx)\nabla_{\bk}^{2}\phi(\bk)\right\rangle \,\mathrm{d}\theta\mathrm{d}s\,\dd\bk' \, \Bigg|.
	\end{align*}
	Using the H\"older inequality, we obtain
    \begin{equation}\label{est:Q-int-error}
	\begin{aligned}
	    \left|\Err^{\mathrm{int}}_{\pm}[\eta, \phi] \right|
		&\lesssim \hbar\,\left\Vert \bk'^{2} \widehat{V}_{\mathrm{int}}(\bk')  \right\Vert _{L^1_{\bk'}} 
		\int_{0}^{1}\!\!\int_{0}^{1}\left\Vert \Tr\left[\Pi_{\pm}\cW^\hbar\left(t,\bx,\bk+\theta(2s-1)\bk'\right)\right]\right\Vert_{L_{\bx}^{1}L_{\bk}^{1}}
		\dd \theta \dd s\;
		\left\Vert \eta(\bx)\right\Vert _{L_{\bx}^{\infty}}\left\Vert \nabla_{\bk}^{2}\phi(\bk)\right\Vert _{L_{\bk}^{\infty}}\\[4pt]
		&\lesssim \hbar\left\Vert  \bk'^{2} \hat{K}(\bk') \hat{\rho}(\bk') \right\Vert _{L^1_{\bk'}}
		\left\Vert \Tr\left[\Pi_{\pm}\cW^\hbar(t,\bx,\bk)\right] \right\Vert_{L_{\bx}^{1}L_{\bk}^{1}}
		\left\Vert \eta(\bx) \right\Vert_{L_{\bx}^{\infty}}
		\left\Vert \phi(\bk) \right\Vert_{W^{2,\infty}_\bk}\\[4pt]
		&\lesssim \hbar \left\Vert \bk' \hat{K}(\bk') \right\Vert _{L^1_{\bk'}} \left\Vert \bk' \hat{\rho}(\bk')  \right\Vert _{L^{\infty}_{\bk'}}
		\left\Vert \Tr\left[\Pi_{\pm}\cW^\hbar(t,\bx,\bk)\right] \right\Vert_{L_{\bx}^{1}L_{\bk}^{1}}
		\left\Vert \eta(\bx) \right\Vert_{L_{\bx}^{\infty}}
		\left\Vert \phi(\bk) \right\Vert_{W^{2,\infty}_\bk}\\[4pt]
		&\leq C\hbar
		\left\Vert \eta(\bx) \right\Vert_{L_{\bx}^{\infty}}
		\left\Vert \phi(\bk) \right\Vert_{W^{2,\infty}_\bk},
	\end{aligned}
    \end{equation}
	where we utilize the fact that $K \in A^1(\R^2_\bx)$ and $\left\Vert \bk' \hat{\rho}(\bk')\right\Vert _{L^{\infty}_{\bk'}} \lesssim  \left\Vert \psi_0 \right\Vert _{H^1_\bx}$ in the last inequality.
	This completes the proof.  
\end{proof}

\subsection{Estimate of total error terms}

Now, we are in a position to complete the proof of Theorem \ref{thm:massiv} and \ref{thm:massless} by collecting the estimates of total errors from Proposition \ref{prop:weak-kinetic}, \ref{prop:error-Vext} and \ref{prop:error-Vint}.

\begin{proof}
	By substituting \eqref{def:kineticpart}-\eqref{def:error-Qint} into \eqref{eq:proj-eqn}, we can formulate the total error terms in the sense that, 
	\begin{equation*}
		\langle  \Err_\pm^{\hbar}(t,\bx,\bk), \eta(\bx) \phi(\bk) \rangle = \Err^{\mathrm{kin}}_\pm[\eta,\phi] + \Err^{\mathrm{ext}}_{\pm}[\eta, \phi] + \Err^{\mathrm{int}}_{\pm}[\eta, \phi].
	\end{equation*}
	
	For the massive case ($m > 0$), collecting the error estimates $\eqref{eq:weak-kinetic-merged}$, \eqref{est:error-Vext} and \eqref{est:error-Vint} from Propositions \ref{prop:weak-kinetic}, \ref{prop:error-Vext} and \ref{prop:error-Vint}, respectively, we can obtain, for test functions $\eta \in W^{1,\infty}_\bx$ and $\phi \in W^{2,\infty}_\bk$,
	\begin{equation*}
		\begin{aligned}
			|\langle  \Err_\pm^{\hbar}(t,\bx,\bk), \eta(\bx) \phi(\bk) \rangle| &\;\leq\;  |\Err^{\mathrm{kin}}_\pm[\eta,\phi]| + |\Err^{\mathrm{ext}}_{\pm}[\eta, \phi]| + |\Err^{\mathrm{int}}_{\pm}[\eta, \phi]| \\[5pt]
			&\;\leq\; C\left[C_T\,\hbar\,\|\eta\|_{W^{1,\infty}_\bx }\|\phi\|_{H^1_\bk}  + \frac{\hbar}{\ve^2}\,\|\eta\|_{L^\infty_\bx}\,\|\phi\|_{W^{2,\infty}_\bk} + \hbar \,\|\eta\|_{L^\infty_\bx}\,\|\phi\|_{W^{2,\infty}_\bk} \right],
		\end{aligned}
	\end{equation*}
	uniformly for $t \in [0,T]$, which directly yields \eqref{est-total-massive}.
	
	For the massless case ($m=0$), we can similarly combine the error estimates  $\eqref{eq:weak-kinetic-merged}$, \eqref{est:error-Vext} and \eqref{est:error-Vint} that, for test functions $\eta \in W^{1,\infty}_\bx$ and $\phi \in W^{2,\infty}_\bk(\Omega_\kappa)$ for any $\kappa>0$,
	\begin{equation*}
		\begin{aligned}
			|\langle  \Err_\pm^{\hbar}(t,\bx,\bk), \eta(\bx) \phi(\bk) \rangle| &\;\leq\;  |\Err^{\mathrm{kin}}_\pm[\eta,\phi]| + |\Err^{\mathrm{ext}}_{\pm}[\eta, \phi]| + |\Err^{\mathrm{int}}_{\pm}[\eta, \phi]| \\[5pt]
			&\;\leq\; C\left[C_T\,\hbar\,\|\eta\|_{W^{1,\infty}_\bx }\|\phi\|_{H^1_\bk(\Omega_\kappa)} + \frac{\hbar}{\ve^2}\,\|\eta\|_{L^\infty_\bx}\,\|\phi\|_{W^{2,\infty}_\bk} + \hbar \,\|\eta\|_{L^\infty_\bx}\,\|\phi\|_{W^{2,\infty}_\bk} \right],
		\end{aligned}
	\end{equation*}
	leading to \eqref{est-totoal-massless}.
\end{proof}

\section*{Acknowledgment}

The authors thank to Chiara Saffirio and Heinz Siedentop for helpful discussions.
J.L. is supported by the Swiss National Science Foundation through the NCCR SwissMAP and the SNSF Eccellenza project PCEFP\_181153, and by the Swiss State Secretariat for Research and Innovation through the project P.530.1016 (AEQUA).
K.Q. is partially supported by the AMS-Simons Travel Award grant, and part of this work is based upon the support by the National Science Foundation under Grant No.~DMS-2424139, while K.Q. was in residence at the Simons Laufer Mathematical Sciences Institute in Berkeley, California, during the Fall 2025 semester.
\appendix

\section{Preliminaries about Dirac equation}\label{Appendix:Basics}
\begin{definition}[Notations for the Dirac equation]
	We denote
	\begin{align*}
		\gamma^0 = \sigma^3, \qquad
		\gamma^1 = \im \sigma^1, \qquad
		\gamma^2 = \im \sigma^2,
	\end{align*}
	where $\sigma^j$ are the Pauli matrices:
	\begin{align*}
		\sigma^1 &= 
		\begin{pmatrix}
			0 & 1 \\
			1 & 0
		\end{pmatrix}, \quad
		\sigma^2 = 
		\begin{pmatrix}
			0 & -\im \\
			\im & 0
		\end{pmatrix}, \quad
		\sigma^3 = 
		\begin{pmatrix}
			1 & 0 \\
			0 & -1
		\end{pmatrix}.
	\end{align*}
	Then, the matrices $\boldsymbol{\alpha} = (\alpha^1, \alpha^2)$ are
	\[
	\alpha^j := \gamma^0 \gamma^j \quad \Rightarrow \quad
	\alpha^1 = \sigma^2, \quad
	\alpha^2 = -\sigma^1.
	\]    
\end{definition}

\begin{lemma}[$L^2_\bx$-Conservation]\label{lem:L2conservation}
	Under the assumptions \textnormal{\textbf{(A1)}} and \textnormal{\textbf{(A2)}}, 
	let $\psi_0\in L^2_\bx$,
	\[
	V(t,\bx)=V_\Gamma\left( \frac{\bx}{\ve}\right) + V_\mathrm{int}(t,\bx),
	\qquad
	V_\mathrm{int}(t,\bx)=\big(K*|\psi^\hbar(t,\cdot)|^2\big)(\bx),
	\]
	and $\psi^\hbar(t)$ be the (mild) solution of \eqref{eq:massive-dirac-scaled} (or \eqref{eq:massless-dirac-scaled}) for any $t \geq 0$. Then, we have
	\[
	\|\psi^\hbar(t)\|_{L^2_\bx}=\|\psi_0\|_{L^2_\bx}.
	\]
\end{lemma}

\begin{proof}
	Let us rewrite \eqref{eq:massive-dirac-scaled} (or \eqref{eq:massless-dirac-scaled}) as $\ii\hbar\,\partial_t^\hbar\psi=(H_0+V)\psi^\hbar$ with
	$H_0:= -\ii \hbar c\,\balpha\cdot\nabla + m c^2\,\gamma^0$ (or $H_0:= -\ii \hbar c\,\balpha\cdot\nabla$, respectively), which is self-adjoint, and $V(t,\bx)$ is a real-valued multiplicative operator.
	Therefore, $H_0+V$ is Hermitian for each fixed $t$.
	
	Then,
	\[
	\frac{\mathrm{d}}{\mathrm{d} t}\|\psi^\hbar(t)\|_{L^2_\bx}^2
	=\frac{2}{\hbar}\,\Im\langle (H_0+V)\psi^\hbar,\psi^\hbar\rangle
	=\frac{2}{\hbar}\Big(\Im\langle H_0\psi^\hbar,\psi^\hbar\rangle+\Im\langle V\psi^\hbar,\psi^\hbar\rangle\Big)=0,
	\]
	since $\langle H_0\psi^\hbar,\psi^\hbar\rangle\in\R$ by self-adjointness of $H_0$, and
	$\langle V\psi^\hbar,\psi^\hbar\rangle=\int_{\R^2} V|\psi^\hbar|^2\,\dd \bx\in\R$ as $V$ is real-valued.
\end{proof}

\begin{lemma}[Bound of $H^1_\bx$-norm]\label{lem:bdd-H1}
	Under the assumptions \textnormal{\textbf{(A1)}} and \textnormal{\textbf{(A2)}},
	let $\psi_0\in H^1_\bx$,
	\[
	V(t,\bx)=V_\Gamma\left( \frac{\bx}{\ve} \right) + V_\mathrm{int}(t,\bx),
	\qquad
	V_\mathrm{int}(t,\bx) = \big(K*|\psi^\hbar(t,\cdot)|^2\big)(\bx),
	\]
	
	then, the solution $\psi^\hbar\in C\big([0,T];H^1_\bx\big) \cap C^1\big([0,T];L^2_\bx\big)$ to
	\[
	\ii\hbar\partial_t\psi^\hbar=H(t)\psi^\hbar \quad  \text{with} \quad 
	H(t)=-\ii\hbar c\,\balpha\!\cdot\!\nabla+mc^2\gamma^0+V(t,\bx)
	\]
	satisfies, for all $t\in[0,T]$,
	\begin{equation}\label{eq:H1-growth-main}
		\|\psi^\hbar(t)\|_{H^1_\bx}
		\;\le\;
		\|\psi_0\|_{H^1_\bx}
		+ C\,t\left(\frac{1}{\ve}\|\nabla V_\Gamma\|_{L^\infty_\bx}
		+\|\nabla K\|_{L^\infty_\bx} \,\|\psi_0\|_{L^2_\bx}^2\right)\,\|\psi_0\|_{L^2_\bx},
	\end{equation}
	where the constant $C>0$ depends only on $c,m$ but not on $\ve$ and $\hbar$.
\end{lemma}

\begin{proof}
	Let 
	\[
	G:=H_0+\ii \quad \text{with} \quad  H_0 =-\ii\hbar c\,\balpha\!\cdot\!\nabla+mc^2\gamma^0.
	\]
	The graph norm $\|G\cdot\|_{L^2_\bx}$ 
	is equivalent to $\|\cdot\|_{H^1_\bx}$ uniformly in $\hbar$.
	Denoting $E_1(t):=\|G\psi^\hbar(t)\|_{L^2_\bx}^2$ and using $\partial_t\psi^\hbar=-(\ii/\hbar)H(t)\psi^\hbar$,
	\[
	\frac{\dd}{\dd t}E_1(t)
	=\frac{2}{\hbar}\Im\langle G\psi^\hbar,[G,V]\psi^\hbar\rangle,
	\qquad
	[G,V]=[H_0,V]=-\,\ii\hbar c\,\balpha\!\cdot\!\nabla V,
	\]
	then, we have $\|[G,V]\|_{\mathcal B(L^2)}\le \hbar c\,\|\nabla V\|_{L^\infty_\bx}$ and
	\[
	\frac{\dd}{\dd t}E_1^{\frac{1}{2}}(t)\le c\,\|\nabla V(t)\|_{L^\infty_\bx}\,\|\psi_0\|_{L^2_\bx}.
	\]
	By Young’s inequality and mass conservation,
	\[
	\|\nablax V(t)\|_{L^\infty_\bx}
	\le \frac{1}{\ve}\|\nablax V_\Gamma\|_{L^\infty_\bx}
	+ \|\nablax K \|_{L^\infty_\bx}\,\||\psi^\hbar|^2\|_{L^1_\bx}
	\le \frac{1}{\ve}\|\nablax V_\Gamma\|_{L^\infty_\bx}
	+ \|\nablax K\|_{L^\infty_\bx}\,\|\psi_0\|_{L^2_\bx}^2.
	\]
	Integrating in time and using the equivalent relation $\|G\psi^\hbar\|_{L^2_\bx} \simeq \|\psi^\hbar\|_{H^1_\bx}$ yields \eqref{eq:H1-growth-main}.
\end{proof}

\begin{lemma}[Property of $\Pi_{\pm}$]\label{lem:proj-derivative}
	Let $\Pi_{\pm}(\bk)$ be the spectral projectors defined in~\eqref{eq:proj-general}.
	Then, the following identities hold:
	\begin{equation}\label{eq:proj-derivative-identities}
		\Pi_\pm(\nabla_{\bk}\Pi_\pm)\Pi_\pm = 0,
		\qquad
		\Pi_\pm(\nabla_{\bk}\Pi_\mp)\Pi_\pm = 0.
	\end{equation}
\end{lemma}

\begin{proof}
	Let $P(\bk)$ be any smooth family of orthogonal projections, in particular $P=\Pi_\pm$.
	Since $P^2=P$, differentiating with respect to $k_j$ gives
	\[
	\partial_{k_j}(P^2)=\partial_{k_j}P
	\quad\Longrightarrow\quad
	(\partial_{k_j}P)P+P(\partial_{k_j}P)=\partial_{k_j}P.
	\]
	Multiplying this relation on both sides by $P$ and using $P^2=P$ yields
	\[
	P(\partial_{k_j}P)P + P(\partial_{k_j}P)P = P(\partial_{k_j}P)P,
	\]
	hence, $P(\partial_{k_j}P)P = 0$.
	Since this holds for each component $k_j$, it follows that
	\[
	P(\nabla_{\bk}P)P = 0.
	\]
	Applying this to $P=\Pi_\pm$ yields the first identity in \eqref{eq:proj-derivative-identities}.
	
	For the second identity, note that $\Pi_- = \Id - \Pi_+$, we have
	\[
	\nabla_{\bk}\Pi_- = -\,\nabla_{\bk}\Pi_+,
	\qquad
	\Pi_\pm(\nabla_{\bk}\Pi_\mp)\Pi_\pm
	= -\,\Pi_\pm(\nabla_{\bk}\Pi_\pm)\Pi_\pm = 0.
	\]
	This completes the proof of \eqref{eq:proj-derivative-identities}.
\end{proof}

\section{Calculation of \texorpdfstring{Wigner transform $\cW^\hbar$}{Wigner transform} in component}
\label{Appendix:Componentwise}

This section is to provide a detailed calculation to obtain \eqref{eq:Wigner} in the component form. Note that, throughout this appendix, we use $\psi^\dagger$ (and the notation
$\overline\psi$) for the usual Hermitian transpose of the spinor,
so that $\overline\psi_j = \overline{\psi_j}$ in all component sense.
In particular, the Wigner matrix $\cW^\hbar$ is the Weyl transform
of the density matrix $\psi^\hbar (\psi^\hbar)^\dagger$ on
$L^2(\R^2;\C^2)$, and \emph{not} of the Dirac adjoint
$\bar\psi^\hbar := (\psi^\hbar)^\dagger \gamma^0$.

Let us write $\psi=(\psi_1,\psi_2)^{\!\top}$ and define the shifted fields
\[
\psi_{\pm,j}^\hbar(t,\bx,\by):=\psi_j\Bigl(t,\bx\pm\tfrac{\hbar}{2}\by\Bigr),\qquad
\overline\psi_{\pm,j}^\hbar(t,\bx,\by):=\psi^\dagger_j\!\Bigl(t,\bx\mp\tfrac{\hbar}{2}\by\Bigr).
\]
With $\gamma^0=\mathrm{diag}(1,-1)$, each component of the Wigner matrix are
\[
W_{11}^\hbar=\frac{1}{(2\pi)^2}\!\int_{\R^2} \ee^{-\ii\bk\cdot\by}\,\psi_{+,1}^\hbar\,\ol\psi_{-,1}^\hbar\,\dd\by,\quad
W_{12}^\hbar=\frac{1}{(2\pi)^2}\!\int_{\R^2} \ee^{-\ii\bk\cdot\by}\,\psi_{+,1}^\hbar\,\ol\psi_{-,2}^\hbar \,\dd\by,
\]
\[
W_{21}^\hbar=\frac{1}{(2\pi)^2}\!\int_{\R^2} \ee^{-\ii\bk\cdot\by}\,\psi_{+,2}^\hbar\,\ol\psi_{-,1}^\hbar \,\dd\by,\quad
W_{22}^\hbar=\frac{1}{(2\pi)^2}\!\int_{\R^2} \ee^{-\ii\bk\cdot\by}\,\psi_{+,2}^\hbar\,\ol\psi_{-,2}^\hbar \,\dd\by .
\]
Furthermore, considering the Dirac equation in the component sense, 
\[
\ii\hbar\partial_t\psi^\hbar = \left(-\ii\hbar c\,\balpha\!\cdot\!\nabla + mc^2\gamma^0 + V\right)\psi^\hbar,
\]
with $V = V_\Gamma + V_{\mathrm{int}}$,
we have
\begin{align*}
	\partial_t\psi_1^\hbar&=-c(\partial_{x_1}-\ii\partial_{x_2})\psi_2^\hbar -\frac{\ii}{\hbar}(mc^2+V)\,\psi_1^\hbar, \quad
	\partial_t\psi_2^\hbar=-c(\partial_{x_1}+\ii\partial_{x_2})\psi_1^\hbar +\frac{\ii}{\hbar}(mc^2 - V)\,\psi_2^\hbar,\\
	\partial_t\ol\psi_1^\hbar&=-c(\partial_{x_1}+\ii\partial_{x_2})\ol\psi_2^\hbar +\frac{\ii}{\hbar}(mc^2+V)\,\ol\psi_1^\hbar,\quad 
	\partial_t\ol\psi_2^\hbar=-c(\partial_{x_1}-\ii\partial_{x_2})\ol\psi_1 -\frac{\ii}{\hbar}(mc^2 - V)\,\ol\psi_2.
\end{align*}
All derivatives on the right are evaluated at the shifted points. And the following key chain-rule identities for the shifts holds:
\[
\partial_{x_j}F\!\Bigl(\bx\pm\tfrac{\hbar}{2}\by\Bigr)=\pm\frac{2}{\hbar}\,\partial_{y_j}F\!\Bigl(\bx\pm\tfrac{\hbar}{2}\by\Bigr).
\tag{CR}
\]

Let us take the $(1,1)$-entry as an example and proceed as in the massless case. The only new term comes from the mass:
\[
\partial_t W_{11}^\hbar
=\frac{1}{(2\pi)^2}\!\int_{\R^2} \ee^{-\ii\bk\cdot\by}
\left[
(\partial_t\psi_{+,1}^\hbar)\,\ol\psi_{-,1}^\hbar
+\psi_{+,1}^\hbar\,(\partial_t\ol\psi_{-,1}^\hbar)
\right]\,\dd\by.
\]
Inserting the componentwise Dirac equations into the definition of $W_{11}^\hbar$ gives
\begin{align*}
	\partial_t W_{11}^\hbar
	&=\frac{1}{(2\pi)^2}\!\int_{\R^2} \ee^{-\ii\bk\cdot\by}\Bigl[
	\bigl(-c(\partial_{x_1}-\ii\partial_{x_2})\psi_{+,2}^\hbar-\tfrac{\ii}{\hbar}(mc^2+V_+)\psi_{+,1}^\hbar\bigr)\,\ol\psi_{-,1}^\hbar \\
	&\hspace{4.5cm}+\psi_{+,1}^\hbar\bigl(-c(\partial_{x_1}+\ii\partial_{x_2})\ol\psi_{-,2}^\hbar + \tfrac{\ii}{\hbar}(mc^2+V_-)\ol\psi_{-,1}^\hbar\bigr)
	\Bigr]\dd\by
\end{align*}
with $V_\pm:= V \left(\bx\pm\tfrac{\hbar}{2}\by \right)$.  
Separating the terms yields exactly the same differential part as in the massless case, together with the potential difference:
\[
\partial_t W_{11}^\hbar
=\text{[massless terms]}
- \frac{\ii}{\hbar}\frac{1}{(2\pi)^2}\int_{\R^2} \ee^{-\ii\bk\cdot\by}(V_+ - V_-)\,\psi_{+,1}^\hbar\,\ol\psi_{-,1}^\hbar\,\dd\by.
\]
The mass contributions $\propto mc^2$ cancel out identically. Indeed,
\[
-\frac{\ii}{\hbar}mc^2\bigl(\psi_{+,1}\,\ol\psi_{-,1}-\psi_{+,1}\,\ol\psi_{-,1}\bigr)=0,
\]
so no explicit mass term remains in $\partial_t W_{11}$ except through the commutator with
$H_m(\bk)=c\,\boldsymbol{\alpha}\cdot\bk+mc^2\gamma^0$.

Thus, the entire calculation proceeds as in the massless case, and the mass term contributes only via a \emph{diagonal commutator}:
\[
H_m(\bk) = c\,\boldsymbol{\alpha}\cdot\bk + mc^2\,\gamma^0.
\]
Other $i,j\in\{1,2\}$ are very similar to this calculation.
Then, the final equation becomes
\[
\partial_t W_{ij}^\hbar + c\,(\balpha\!\cdot\!\nabla_{\bx}W^\hbar)_{ij}
- \big(\mathcal Q[V,W^\hbar]\big)_{ij}
+ \frac{1}{\ii\hbar}[H_m(\bk),W^\hbar]_{ij} = 0,
\]
where the Hamiltonian now includes the mass:
\[
H_m(\bk) := c\,\boldsymbol{\alpha}\cdot\bk + mc^2\,\gamma^0.
\]

Therefore, the structure of the Wigner equation remains the same:
\begin{equation}\label{eq:Wigner-massive}
	\partial_t \cW^\hbar
	+ c\,\boldsymbol{\alpha}\cdot\nabla_{\bx} \cW^\hbar
	- \mathcal Q[V, \cW^\hbar]
	+ \frac{1}{\ii\hbar}[H_m(\bk),\,\cW^\hbar] = 0,
	\qquad H_m(\bk)=c\,\boldsymbol{\alpha}\cdot\bk + mc^2\gamma^0.
\end{equation}

This equation reduces to the massless case as $m \to 0$. The expressions of the band projector will also change accordingly, and eigenvalues of $H_0(\bk)=c\,\boldsymbol{\alpha}\cdot\bk$ become $\pm E_0(\bk) := \pm \,c\,|\bk|$.

\bibliographystyle{abbrv}
\bibliography{refs}

\end{document}